\documentclass[onecolumn,nofootinbib,superscriptaddress]{revtex4-2}
\usepackage{graphicx} 
\usepackage{amsfonts}
\usepackage{bbm}
\usepackage{amsthm}
\usepackage{amsmath}
\usepackage{xcolor}
\usepackage{amssymb}
\usepackage{svg}
\usepackage{bm}
\usepackage{quantikz}
\usepackage{ulem} 
\usepackage[filecolor=blue]{hyperref}
\usepackage{cleveref}
\usepackage{algorithm}
\usepackage{algorithmic}
\usepackage{paralist}
\usepackage{quantikz}
\usepackage{subcaption}
\usepackage{float} 

\usepackage{titlesec}

\titleformat{\paragraph}[runin]
  {\normalfont\bfseries}
  {\theparagraph}
  {1em}
  {}

\titlespacing*{\paragraph}
  {0pt}
  {3.25ex plus 1ex minus .2ex}
  {1em} 

\usepackage[table]{xcolor}
\usepackage[filecolor=blue]{hyperref}

\usepackage{braket}

\theoremstyle{plain}
\newtheorem{theorem}{Theorem}[section]

\newtheorem{definition}{Definition}[section]

\begin{document}

\title{Machine learning with minimal use of quantum computers: Provable advantages in Learning Under Quantum Privileged Information (LUQPI)}
\author{Vasily Bokov}
\affiliation{$\langle aQa^L\rangle$, Leiden University, The Netherlands}
\affiliation{LIACS, Leiden University, Niels Bohrweg 1, 2333 CA, Leiden, The Netherlands}
\affiliation{Honda Research Institute Europe GmbH, Carl-Legien-Str.\ 30, 63073 Offenbach, Germany}

\author{Lisa Kohl}
\affiliation{Cryptology Group, CWI Amsterdam, The Netherlands}

\author{Sebastian Schmitt}
\affiliation{Honda Research Institute Europe GmbH, Carl-Legien-Str.\ 30, 63073 Offenbach, Germany}

\author{Vedran Dunjko}
\affiliation{$\langle aQa^L\rangle$, Leiden University, The Netherlands}
\affiliation{LIACS, Leiden University, Niels Bohrweg 1, 2333 CA, Leiden, The Netherlands}

\begin{abstract}

Quantum machine learning (QML) is often listed as a promising candidate for useful applications of quantum computers, in part due to numerous proofs of possible quantum advantages. A central question is how small a role quantum computers can play while still enabling provable learning advantages over classical methods. We study an especially restricted setting where a quantum computer is used \textit{only} as a feature extractor: it acts independently on individual data points, without access to labels or global dataset information, and furthermore assumed to be available only as means to augment the training set, and is not available in deployment.  In other words, the training and deployment are carried out by fully classical learners on a dataset augmented with these quantum-generated features.

We formalize this model by adapting the classical framework of Learning Under Privileged Information (LUPI) to the quantum case, which we call Learning Under Quantum Privileged  Information (LUQPI).
Within this framework, we show that even such minimally involved quantum feature extraction available only for the training data can nonetheless yield exponential quantum-classical separations for suitable concept classes and distributions, under reasonable computational assumptions.
We further situate LUQPI in a taxonomy of related quantum and classical settings and also show how standard classical machinery---most notably the SVM+ algorithm---can exploit quantum-augmented data. 
In the latter direction, we present numerical experiments in a physically motivated many-body setting, where privileged quantum features are expectation values of observables on ground states, and observe consistent performance gains for LUQPI-style models over strong classical baselines.

\end{abstract}

\maketitle

\section{Introduction}

Quantum machine learning (QML) is often regarded as a promising area where quantum computers might eventually outperform classical information processing.
Over the past decade, the field has undergone rapid empirical and theoretical development, and a number of results now suggest that, at least in principle, quantum models can outperform classical learners on carefully designed tasks \cite{Havlicek2019, Gyurik_sep, Jerbi2024}. 

Beyond establishing that \textit{some} quantum advantage is possible, an intriguing question is: what is the \textit{minimal} way in which a quantum computer needs to be involved in the learning pipeline in order to obtain such an advantage?
Recent work has demonstrated provable quantum speed-ups in scenarios where a quantum computer performs the entire training procedure while the deployment is classical \cite{Jerbi2024}, as well as settings where just the inference or evaluation stage is quantum \cite{XanaduBowles, kurkin}.

In this work, we push this ``minimal involvement'' perspective much further.
We ask whether one can obtain provable advantages when the quantum computer is used only as a \textit{feature extractor}: it acts on each individual input datapoint and has no access to labels nor any global property of the dataset.
In particular, the quantum device never sees training labels and never performs end-to-end optimization; it is only allowed to compute additional features that are then handed to a classical learner in an `augmented' dataset.

Conceptually, our approach is inspired by the perspectives of (quantum)  topological data analysis (TDA) \cite{Lloyd2016,LLSD, Berry2024}: there as well, a (typically expensive) procedure extracts features from unlabeled point clouds, after which more standard learning methods are applied.
In our setting, the ``expensive procedure'' is a quantum algorithm acting on individual inputs, and the extracted features are quantities that are believed to be hard to compute classically.
The hope is that these features reveal a structure that makes downstream classical learning much easier, while keeping the quantum role as restricted as possible.

Formally, the model we study aligns with the classical framework of Learning Under Privileged Information (LUPI) \cite{Vapnik2009}, in which, during training, a learner has access to additional information that will not be available at deployment time for new, to-be-labeled, points.
Our work can be understood as an instantiation of this framework with an efficient quantum algorithm that analyzes data points in an i.i.d.\ fashion, and hence we refer to our setting as Learning Under Quantum Privileged Information (LUQPI).

We highlight that LUQPI constitutes a significant restriction compared with previous approaches to quantum-enhanced learning, many of which either allow the quantum algorithm to access labels, to process multiple data points jointly, or to control the full training loop.
Under our constraints, the quantum device cannot directly uncover correlations between inputs and outputs, and in particular cannot itself carry out supervised learning.  We discuss this later in more detail.

The contributions of this paper are as follows.
\begin{itemize}
    \item[(i)] We formally define advantageous learning scenarios with quantum feature extraction and introduce the two natural versions:  \textit{quantum online} - where the feature extraction is available in the inference step, and our key model: the \textit{offline} version that is LUQPI. 
   \item[(ii)] We prove that even in this substantially constrained LUQPI setting, one can construct concept classes for which quantum feature extraction enables exponential advantages over any efficient classical learner under reasonable complexity-theoretic assumptions. The advantage even holds against non-uniform\footnote{In this work, we will emphasize the importance of learning advantages over uniform versus non-uniform learners, corresponding to BPP versus P/poly computational classes. In parallel, we will also distinguish scenarios where the distribution over the input points is uniform (or special/contrived). We emphasize this to minimize the chance of confusion, e.g., assumptions that ``uniform learners'' (or ``non-uniform'' learners) pertains to learners relative to the uniform (or non-uniform) input distribution; in this case uniformity (or non-uniformity) does not refer to the distribution, but to the fact that the boolean circuits defining the learner for each size have efficient Turing machines that generate them (or are defined by a (computationally unbounded) advice string. } learners, that is, classical learners that, in addition to a dataset, are also given additional polyonmially-sized advice, which depends on the size of the learning task at hand.
    \item[(iii)] We provide extensive numerical experiments for a physically motivated problem where quantum-privileged information can be computed from ground states of many-body systems.
    These experiments show that providing privileged features during training can, in certain cases, improve the performance of classical learners, even when those features are unavailable at deployment (i.e. in the test phase).
\end{itemize}

\section{Background}
\label{sec:background}

To put our results on a firm footing, we briefly review the Probably Approximately Correct (PAC) learning framework and introduce our notation.
We then discuss a general notion of feature extraction and how it interacts with standard PAC definitions.

\subsection{PAC learning}

For each input size \(n \in \mathbb{N}\), let \(\mathcal{X}_n\) be a domain of instances and let \(\mathcal{Y}\) be a label space (for example, \(\{0,1\}\) for classification or $\mathbb{R}$ for regression).
A \textit{concept class} is a family \(\mathcal{C} = \bigcup_{n \ge 1} \mathcal{C}_n\), where each \(\mathcal{C}_n \subseteq \mathcal{Y}^{\mathcal{X}_n}\) consists of (Boolean or real-valued) functions \(c : \mathcal{X}_n \to \mathcal{Y}\).

Analogously, it will be expedient for us to define the \textit{distribution class}, which is a family \(\mathcal{D} = \bigcup_{n \ge 1} \mathcal{D}_n\), where each \(\mathcal{D}_n\) constitutes a set of discrete distributions over \(\mathcal{X}_n\) (e.g., $\{0,1 \}^n$).

The learning task is to (approximately) identify an unknown target concept \(c \in \mathcal{C}_n\) from labeled examples, where the data follows one of the distribution class distributions.

A \textit{learning algorithm} is given i.i.d.\ labeled examples \((x, c(x))\) where the inputs \(x\) are drawn from a distribution \(\mathcal{D}_n\) over \(\mathcal{X}_n\) (this distribution can be fixed or arbitrary depending on the setting).
Equivalently, we can view the learner as having sample access to an \textit{example oracle} \(EX(c,D_n)\) that, upon each call, returns a fresh pair \((x, c(x))\) with \(x \sim D_n\).
Based on a finite sample of size \(m\), the learner outputs a hypothesis \(h : \mathcal{X}_n \to \mathcal{Y}\).
In the remainer of this document, we will assume the domains and codomains are bitstrings (or integers), unless otherwise specified.

\begin{definition}[Efficient (classical and quantum) PAC learnability]
\label{def:pac}
A concept class \(\mathcal{C} = \bigcup_n \mathcal{C}_n\) is \textit{efficiently PAC learnable} relative to the distribution class  $\mathcal{D} =\bigcup_{n \ge 1} \mathcal{D}_n $ if there exist a learning algorithm \(\mathcal{A}\) and a polynomial \(p\) such that for every \(n\), every \(c \in \mathcal{C}_n\), every distribution \( D \in \mathcal{D}_n\) from the family and all precision and confidence parameters \(0<\epsilon,\delta<1/2\), the following holds:
given on input \(m = p(n,1/\epsilon,\log(1/\delta))\) samples from \(EX(c,\mathcal{D}_n)\), and precision parameters $\epsilon$, $\delta$, 
the algorithm \(\mathcal{A}\)  
runs in time at most \(p(n,1/\epsilon,\log(1/\delta))\) and outputs a boolean function (hypothesis) \(h\) from a hypothesis class $\mathcal{H}$ such that
\begin{align}
\Pr_{x\sim \mathcal{D}_n}\big[h(x) \neq c(x)\big] \le \epsilon \label{PAC-cond}
\end{align}
with probability at least \(1-\delta\) over the randomness of the sample and the internal randomness of \(\mathcal{A}\).
We say that \( (\mathcal{C}, \mathcal{D}) \) is \textit{classically efficiently learnable} if the above holds for a classical polynomial time uniform\footnote{We will discuss uniform and non-uniform versions of learners shortly.} algorithm $\mathcal{A}$ and where each hypothesis is a boolean circuit of size $O(p(n,1/\epsilon,\log(1/\delta)))$ and \textit{quantumly efficiently learnable}, if the same holds for a quantum algorithm $\mathcal{A}$, and the hypothesis class consists in (randomized) functions computable using quantum circuits of size $O(p(n,1/\epsilon,\log(1/\delta)))$.
\end{definition}

In general, one can consider settings where either the hypothesis class or the algorithm $\mathcal{A}$ is classical or quantum, but here we focus on these fully classical and fully quantum cases, see \cite{Gyurik_sep}.

\paragraph{Learning advantage/separation} We say a PAC learning problem, specified by the pair $(\mathcal{C}, \mathcal{D}) $ exhibits a classical quantum learning separation (or: a quantum learning advantage) if it is quantumly efficiently learnable but not classically efficiently learnable.

Throughout this work, when we speak of (efficient) PAC learnability in a ``canonical'' sense, which means that the learning condition in \ref{PAC-cond}
is achieved for any pairing of $n$-bit concept from $\mathcal{C}$ relative to any $n-$bit distribution from  $\mathcal{D}$.

Later, we will also mention deviations where we only demand the learning condition to be met for some subset of possible pairs, which we call \textit{concept-specific-distribution} setting, which we do not consider strict PAC learning. 

Two special cases of PAC are the ``basic" PAC, where $\mathcal{D}$ contains all possible distributions, and the fixed-distribution PAC, where there is exactly one known distribution per size, say the uniform distribution.

\paragraph{Learning advantage relative to non-uniform learners} In the basic definition above, we assume $\mathcal{A}$ is a \textit{uniform} algorithm, i.e. representable as a poly-time randomized Turing machine\footnote{Or, equivalently, that there exists a poly-time Turing machine which outputs a description of the circuit for each input size $n$}. However, we will also be interested in stronger separations, where we allow non-uniform classical learning algorithms, $i.e.,$ where we only demand that $A$ can be executed using a $O(\textup{poly}(n,1/\epsilon,\log(1/\delta)))$-sized boolean circuit\footnote{Or, equivalently, where the corresponding Turing machine is additionally given a polynomially-sized advice string depending on the input size alone}.
Separations relative to non-uniform classical learners are appealing as they also include algorithms $A$ which are themselves optimized by training on other, related learning tasks, and thus make particular sense in machine learning contexts.
To minimize chances of confusion, we again highlight that uniform (non-uniform) learners pertain to the nature of the learning algorithm, and not the distribution relative to which the learner operates.

\subsection{Feature extraction}

Informally, a \textit{feature extraction} procedure converts raw inputs into representations that highlight salient patterns relevant for the learning task. While in general feature extraction also may imply removing redundancies from data to ease learning, here we are only interested in extracting additional information which is added to the raw data-points.
In our setting, we will be particularly interested in feature extraction procedures that can be implemented by quantum algorithms.

Formally, for each \(n\) we consider a feature space \(\mathcal{F}_n\) and a mapping
\begin{align}
E_n : \mathcal{X}_n \to \mathcal{F}_n.
\end{align}
We denote by \(\mathcal{E} = \{E_n\}_{n\in\mathbb{N}}\) a family of such feature extractors, one for each input dimension.
Given \(\mathcal{E}\), to connect to standard PAC learning formalisms, we introduce an \textit{extended} example oracle that augments each labeled example with its feature vector.

\begin{definition}[Extended example oracle]
For a target concept \(c \in \mathcal{C}_n\), distribution \(D_n\) on \(\mathcal{X}_n\), and feature-extractor family \(\mathcal{E} = \{ E_n \}_{n \in \mathbb{N}}\), the \textit{extended example oracle}
\begin{align}
EX_{\mathrm{ext}}(c,D_n,E_n)
\end{align}
returns i.i.d.\ samples of the form \(((x,E_n(x)),c(x))\) with \(x\sim D_n\).
\end{definition}

In some of the settings we will consider (``online settings''), the learner may then choose hypotheses that act on both the original inputs and the extracted features.
To capture this, we formalize the notion of the \textit{derived} or \textit{effective} hypothesis class induced by \(\mathcal{E}\).

\begin{definition}[Derived/effective hypothesis class induced by \(\mathcal{E}\)]
\label{def:derived-class}
For each \(n\), let \(\mathcal{H}^{\mathcal{E}}_n \subseteq \mathcal{Y}^{\mathcal{X}_n\times\mathcal{F}_n}\) be a hypothesis class defined on the extended space \(\mathcal{X}_n \times \mathcal{F}_n\).
The \textit{derived} class on \(\mathcal{X}_n\) is
\begin{align}
\mathcal{H}^{\mathcal{E}\Rightarrow}_n
\;:=\;
\Big\{ h^\Rightarrow:\mathcal{X}_n\to\mathcal{Y}\ \Big|\ \exists\,h\in\mathcal{H}^{\mathcal{E}}_n\ \text{such that }\
h^\Rightarrow(x)=h\big(x, E_n(x)\big)\ \ \forall x\in\mathcal{X}_n \Big\}.
\end{align}

Given a hypothesis $h$ from the extended space, we refer to the corresponding $h^\Rightarrow$ as the derived hypothesis. 
\end{definition}

We note that in the settings we will consider, $h$ will be classically tractable, whereas $h^\Rightarrow$ won't necessarily be as it involves the evaluation of the feature map.

In these online models, the feature extractor is available both during training and deployment: the augmented pairs \((x, E_n(x))\) can be formed not only for the training examples but also for any new input at test time. In off-line models (LUQPI), the augmented pairs are only available for the training set \footnote{One can also consider an even weaker version we call ``semi-supervised privileged'' settings, where the training set has two parts: a part which provides augmented pairs \( \{(x, E_n(x)) \} \)  (without matching labels!) and a part which provides labeled pairs\( \{(x', y) \} \), where the datapoints ($x$) in the two sets can be fully disjoint. Interestingly, as we mention later, even this even weaker version allows for some type of advantage.}.
This distinction between training-only and training-and-deployment access to features will play an important role in our taxonomy of scenarios provided later.

In the main class of settings we will consider, we will require that both learning and the use of features are classically computationally efficient. In particular, this will imply that the feature extractor \(E_n\) must produce outputs of size polynomial in \(n\). However, evaluating \(E_n(x)\) for any input \(x \in \mathcal{X}_n\) will only be required to be achievable in time polynomial in \(n\) on a quantum computer.

\subsection{Quantum-advantageous feature extraction: online and offline}
\label{subsec:online-offline}

We now formalize the criteria under which a (quantum) feature extractor provides a genuine learning advantage.
Intuitively, a feature-extractor family \(\mathcal{E}\) is advantageous if, without access to the features, no efficient classical learner can solve the target learning problem, while access to the features makes it efficiently learnable. Such settings are captured by the following definition, first for the ``online'' case.

\begin{definition}[Quantum online advantageous feature extraction]
\label{def:online-adv}
Let \(\mathcal{C}=\bigcup_{n}\mathcal{C}_n\) be a concept class over \(\mathcal{X}\) with label space \(\mathcal{Y}\), and let
\(\mathcal{D}=\{\mathcal{D}_n\}_{n}\) be a family of sets of target distributions over \(\mathcal{X}\) (i.e., the distribution class).
A feature-extractor family \(\mathcal{E}=\{E_n:\mathcal{X}_n\to\mathcal{F}_n\}_n\) is said to be \textit{online-advantageous} for \(\mathcal{C}\) under \(\mathcal{D}\) if:
\begin{enumerate}
\item[(Hardness)] \(\mathcal{C}\) is not efficiently PAC learnable by any polynomial-time classical learner under \(\mathcal{D}\).
\item[(Learnability with deploy-time features)]
There exists a (randomized) learner \(\mathcal{A}\), which together with an extended hypothesis class (relative to the feature map family) \(\mathcal{H}^{\mathcal{E}}=\{\mathcal{H}^{\mathcal{E}}_n\}_n\) with
\(\mathcal{H}^{\mathcal{E}}_n \subseteq \mathcal{Y}^{\mathcal{X}_n\times\mathcal{F}_n}\)  for every \(n\), every \(c\in\mathcal{C}_n\), and all \(0<\epsilon,\delta<1/2\) satisfies:
\begin{itemize}
\item Given \(m=\mathrm{poly}(n,1/\epsilon,\log(1/\delta))\) samples from \(EX_{\mathrm{ext}}(c,\mathcal{D}_n,\mathcal{E})\),
      \(\mathcal{A}\) runs efficiently and outputs \(h\in\mathcal{H}^{\mathcal{E}}_n\).
\item With probability at least \(1-\delta\),
      \begin{align}
      \Pr_{x\sim\mathcal{D}_n}\!\big[\, h^\Rightarrow(x)\neq c(x) \,\big]\ \le\ \epsilon,
      \end{align}
      where \(h^\Rightarrow\) is the derived hypothesis (as in Definition~\ref{def:derived-class}) associated with \(h\) from the previous bullet point.
\item The algorithm \(\mathcal{A}\) and the evaluation of \(h(x,x')\) is classically efficient when $x' = E_n(x)$, (and thus, the evaluation of $h^\Rightarrow(x)$ is efficient, given the value of $E_n(x)$).
\end{itemize}
\end{enumerate}
\end{definition}

In words, in the online setting the feature extractor is available both at training and at deployment: whenever the learned predictor is evaluated on a fresh input \(x\), the features \(E_n(x)\) are recomputed (possibly by a quantum device) and fed into the classical model like it is shown on Fig.~\ref{fig:online_offline_comparison}.

\begin{figure}[h]
\centering
\includegraphics[width=0.65\linewidth]{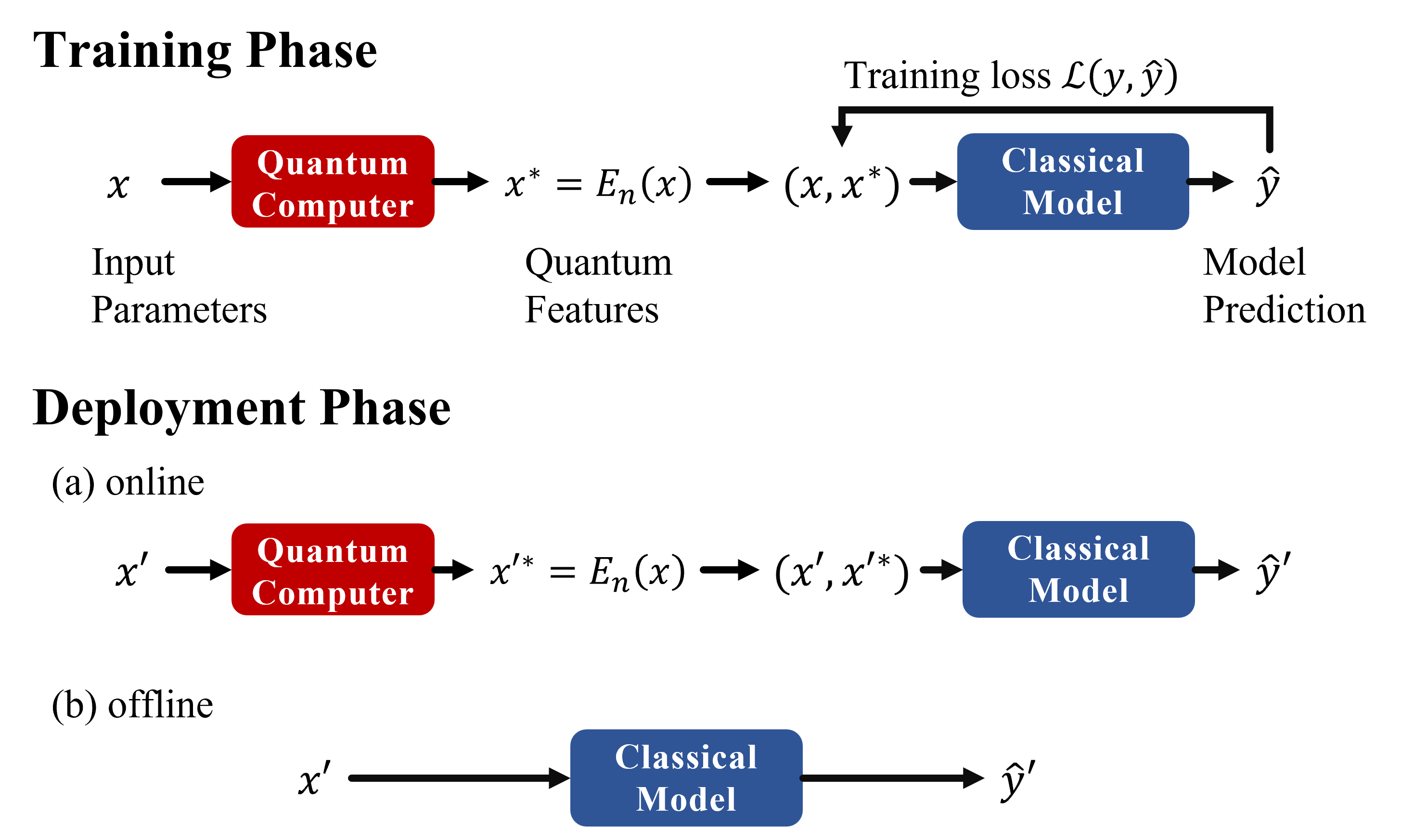}
\caption{Comparison of online and offline feature extraction settings. The training phase is identical for both settings, using quantum feature extraction. In deployment, an online setting requires quantum resources while an offline setting uses purely classical computation. Raw input $x'$ is passed directly to the trained model along with extracted features.}
\label{fig:online_offline_comparison}
\end{figure}

Although we will later discuss the relationship of these notions to prior work, two aspects are important already here:
(i) the quantum subroutine (the feature extractor) only takes individual points as input, and never sees labels, so it cannot itself ``identify correlations'' between many examples and labels; in particular, it cannot directly uncover any supervised signal;
(ii) the learning algorithm that uses these features is entirely classical.
As we will show, already prior works imply separations in this case are nonetheless possible, at least for special distributions; we will improve on this.

We will, however, be more interested in an even more restrictive modality, which we call \textit{quantum-\textbf{offline} advantageous feature extraction}.
Here, the quantum device is used only once, in a data-preprocessing stage, to compute features for the training set.
There is \textit{no} access to the feature extractor at deployment time.

\begin{definition}[Quantum-offline advantageous feature extraction]
\label{def:offline-adv}
With the same setup, \(\mathcal{E}\) is \textit{offline-advantageous} for \(\mathcal{C}\) under \(\mathcal{D}\) if:
\begin{enumerate}
\item[(Hardness)] \(\mathcal{C}\) is not efficiently PAC learnable by any polynomial-time classical learner under \(\mathcal{D}\).
\item[(Learnability \underline{without} deploy-time features)] 
There exist a hypothesis class \(\mathcal{H}^{\mathcal{E},\mathrm{off}}=\{\mathcal{H}^{\mathcal{E},\mathrm{off}}_n\}_n\) and a learner \(\mathcal{A}_{\mathrm{off}}\) such that for every \(n\), every \(c\in\mathcal{C}_n\), and all \(0<\epsilon,\delta<1/2\):
\begin{itemize}
\item Given \(m=\mathrm{poly}(n,1/\epsilon,\log(1/\delta))\) samples from \(EX_{\mathrm{ext}}(c,\mathcal{D}_n,\mathcal{E})\),
      \(\mathcal{A}_{\mathrm{off}}\) runs efficiently and outputs \(\tilde h\in\mathcal{H}^{\mathcal{E},\mathrm{off}}_n\).
\item With probability at least \(1-\delta\),
      \begin{align}
      \Pr_{x\sim\mathcal{D}_n}\!\big[\, \tilde h(x)\neq c(x) \,\big]\ \le\ \epsilon,
      \end{align}
      and \(\tilde h\) is efficiently computable on a classical computer (and in particular computing $E_n$ on a new point $x$ is not needed).
\end{itemize}
\end{enumerate}
\end{definition}

In LUQPI constructions, the quantum device is restricted to this extremely limited role, yet as we show shortly, this still suffices to obtain strong separations.

\section{Taxonomy of scenarios}
\label{sec:taxonomy}

We next relate our setting to existing separations between classical and quantum learning and to other hybrid classical--quantum architectures, leading to a taxonomy of scenarios.

\paragraph{Direct cryptographic approaches} Arguably the oldest provable separations between classical and quantum learners are not far from settings which satisfy the conditions of our framework. In \cite{servedio2000quantumversusclassicallearnability, kearns1994introduction} the authors introduce a concept class based on computing the discrete cube root relative to a modulus $N$ , where $N$ is a 3-RSA\footnote{We call $N=p \times q$ a 3-RSA integer if $p$ and $q$ are odd primes and 3 does not divide $(p-1)$ nor $(q-1)$.} integer enumerating the concepts. The learning of this class is hard, under the so-called discrete cube root assumption (DCRA), see \cite{kearns1994introduction} (roughly, the assumption that computing the cube root is hard on average given $(x,N)$ on input).
 We defer more details of constructions for later, but the key idea toward the separation is that the quantum learner can factor $N$ \cite{Shor_1997}, which allows for an easy solution of the cube root\footnote{While the capacity to factor implies the solving of DCRA, the converse is an open question. }. What is specific for this class is that \textit{for a fixed $N$}, the computation of the discrete cube root can be done by a polynomially-sized classical circuit, as the cube root can be expressed as modular exponentiation with an exponent which depends on the factors of $N$ (see \cite{kearns1994introduction} or \cite{modexp}). This leads to a possibility of quantum-offline settings. However, when one works out all the details and tries to recast such constructions into our offline feature-extraction framework, several choices must be made, specifically, how the modulus $N$ is made known to the learner within the PAC model, which is necessary for quantum learning to be possible.

If \(N\) is treated as a fresh random input for each instance, then the offline possibility disappears, as DCR are known to be $P/poly$ only in the case where the moduli $N$ are fixed per size, and so a new quantum computation would be required for each new point in the deployment/test phase.
One could try to circumvent this problem by allowing $N$ to be fixed for each concept but different between concepts, and given as a part of the output label to the learner. But this no longer fits into a pure feature-extraction formalism as the quantum computer needs to see the labels and even multiple inputs at the same time.
One can consider fixing just one $N$ per size (so all concepts use the same modulus) by somehow specifying the sequence of primes, one per bitstring length. However, to obtain classical intractability, one then requires unconventional and possibly unlikely cryptographic assumptions. As discussed later, in this case, the separation also cannot hold against non-uniform classical learners, who can obtain the factors of $N$ as advice. These options are discussed in detail in see \cite{Gyurik_sep}(section 3.2). 
Here we introduce a fourth option, which is ultimately not satisfactory, but could be of broader interest. To consider pairings of concepts and distributions within the concept class: so each concept $c_N$ comes with its distribution $D^N$ which ``leaks'' information about \(N\). Here, there is only one $N$ in the dataset and the testing/deployment phase, so a classical offline solution is possible, and also the feature extraction format is feasible as each individual datapoint leaks $N$. 
This is what we call \textit{concept--distribution specific} (CDS) PAC learning, which may be reasonable in some applications. However, this is manifestly not ``canonical'' PAC learning and is somewhat unsatisfactory from a learning-theoretic viewpoint and has clear practical limitations. 

In \cite{Jerbi2024}, similar ideas were used to prove learning advantages where just the training is quantum. The construction relied on the discrete cube root problem, similarly to the settings discussed above.
The constructions there did not explicitly specify how the modulus is provided (i.e., the formalism was not exactly PAC). However, no matter how this is resolved: by leaking the modulus, providing it as input, or positing a hard sequence, we end up with something less than desired: label-sensitive settings which see multiple datapoints at the same time, quantum-online settings, or a scenario which requires very strong uncommon assumptions and which does not offer an advantage over non-uniform learners. We will discuss this case and its shortcomings in more detail shortly.
We note that related obstacles will also prevent the modular exponentiation class from \cite{Gyurik_sep} from satisfying the conditions of a quantum-offline feature extraction advantage.

\paragraph{Quantum kernels}
In other directions, prominent lines of works revolve around quantum support vector machines (SVM) and related quantum kernel methods \cite{Liu2021, Schuld2019}.
In SVM-like approaches, the datapoints are embedded into an exponentially large Hilbert space and then processed by (quantum) versions of margin-based classifiers.
This setting does not fit our feature-extraction framework because the learning algorithm itself is quantum and repeatedly queries the state preparation procedure at training and deployment time.
In addition, the feature vectors are elements of a Hilbert space that is not directly accessible as a classical vector, so not a valid feature extraction map.

However, quantum kernel methods, which do not map datapoints to feature spaces explicitly, are more closely related to our setting: the kernel can be interpreted as computing inner products between feature vectors, which in principle correspond to some implicit feature map.
Promisingly, the kernel function does not need to see the labels.
However, this reading still violates at least two of our constraints. First, the kernel acts on \textit{two} datapoints, which feature extraction mechanisms cannot do. Furthermore, this feature map will need to be accessed during deployment, leading to an online setting.

\paragraph{Quantum Extreme Learning Machines (QELM)}
This quantum reservoir computing approach \cite{Xiong_2025} (and also the special case known as the quantum extreme learning machine \cite{QELM}) represents a different implementation of quantum feature extraction. It operates in an online rather than offline deployment setting. Classical input data is first encoded into a quantum state on so-called accessible qubits \cite{QELM}, then evolved under fixed reservoir dynamics involving both accessible and hidden qubits to generate entangled quantum states. Feature extraction is realized through measurements of a predetermined set of observables on the reservoir state, yielding expectation values that serve as classical features. These quantum-derived features are then processed by classical linear regression with trainable weights. The approach is label-independent—quantum features are computed without knowledge of labels—and produces polynomial-dimensional classical feature vectors suitable for efficient processing. However, as mentioned, the quantum feature computation is required for each input during both training and deployment, placing this firmly in the ``online'' category of our classification. 

\paragraph{Other}
Topological data analysis (TDA) provides another instructive comparison point.
There, features are computed from point clouds (or more general combinatorial structures) that summarize topological properties of the data.
At the face of it, this is not a feature map acting on each point individually, unless we take the point-clouds themselves to constitute individual datapoints.
However, even in this case, and even if there was a rigorous proof that these features are hard to compute classically but are tractable for quantum computers (at present this is still a conjecture)\cite{Berry2024}, this would not suffice for the type of separation claim we desire.
In particular, one would still have to prove the existence of non-trivial learning tasks which necessitate the use of these features, and which allows quantum-offline modalities\footnote{One could trivially take these assumed hard topological features to be the desired label/output (for which they would not only need to be outside of $BPP$ machines but also $BPP/samp$ machines), but in this case, we have no approach toward an offline mode.}. 

\vspace{0.3cm}

Below we provide Table ~\ref{tab:qml_classification}, which classifies different quantum learning approaches according to three key properties:

\begin{itemize}
    \item \textbf{Valid feature extraction:} Does the approach define a proper feature extraction map (processing single datapoints without labels), or does it violate these constraints by requiring multiple datapoints, label access, or producing non-classical outputs?
    \item \textbf{Offline deployment:} Are quantum resources needed only during training, or also during inference?
    \item \textbf{Advantage}: Is there a formal proof of advantage/separation for any PAC learning problem relative to plausible assumptions? And if yes, is it relative to uniform or also \textbf{non-uniform} classical learners?
    \item \textbf{Distribution:}  Does the setting rely on sets of contrived distributions? Or, is the distribution natural (e.g., uniform over bitstrings).
\end{itemize}

The desirable characteristics are given in boldface in the table.

\begin{table}[h]
\centering
\small
\begin{tabular}{|l|l|l|l|l|}
\hline
\textbf{Approach} & \textbf{Valid FE?} & \textbf{Offline?} & 
\textbf{Advantage} & \textbf{Distribution}  \\
\hline

Modular exp. & No (sees labels) & \textbf{Yes} & \textbf{Non-uniform} & \textbf{Natural} \\
\hline
Quantum kernels & No (acts on 2 points) & No & Uniform  & \textbf{Natural}\\
\hline
QELM & \textbf{Yes}  & No & Unknown & Unknown \\
\hline
TDA & Unclear (acts on clouds) & Unknown & Unknown & Unknown \\
\hline
DCR / Shadows of QML (1) & \textbf{Yes} & \textbf{Yes}   & \textbf{Non-uniform} & {Contrived/CDS} \\
\hline
DCR / Shadows of QML (2) & \textbf{Yes} & \textbf{Yes}   & Uniform   & \textbf{Natural} \\
\hline
\textbf{Our LUQPI} & \textbf{Yes} & \textbf{Yes} & \textbf{Non-uniform}  &\textbf{Natural} \\
\hline
\end{tabular}

\caption{Classification of quantum learning approaches by feature extraction validity, deployment strategy, guarantees, and naturality of distribution.}
\label{tab:qml_classification}
\end{table}

This Table \ref{tab:qml_classification}  lists both general methods and learning problems where learning separations were exhibited.

As clarified, the constructions from \cite{Jerbi2024} (Shadows of QML) are not technically PAC, but can be modified in which case they are essentially the same as some of the discrete cube root constructions in \cite{Gyurik_sep}, so they are presented together. DCR-based constructions provide a number of options, and the closest one to our objective of quantum-offline settings is discussed shortly in more detail. However, this approach falls short in terms of the strength of guarantees of separation (advantage), and the aturality of distributions.

 Quantum kernels are a general method, however they were used to exhibit a separation in \cite{Q_kernels}. 
For QELM and TDA, to our knowledge, no explicit PAC learning tasks which are solved by these methods have been identified in the literature. We believe they could be constructed; however, the constructions we see likely still do not lead to a quantum-offline advantageous quantum feature extraction setting.

``DCR / Shadows of QML (1)'' case refers to the construction where we couple concepts with distributions, i.e. the concept-distribution specific (CDS) setting, and the distribution leaks the concept-defining modulus $N$ in each input.

As noted, for the case of quantum kernels and the second reading of DCR (``DCR / Shadows of QML (2)''), to achieve a guarantee relative to a uniform distribution, we have to employ a very strong, and likely implausible assumption about a hard sequence of 3-RSA integers, discussed later and in \cite{Gyurik_sep}. However, there is no possibility, under any assumption, to achieve an advantage against non-uniform classical learners for a natural (uniform) distribution.
In contrast, in the new LUQPI construction, we employ a non-standard, however, as we explain, highly plausible assumption, using which we achieve an advantage relative to non-uniform learners as well. 

At this point, we notice that one could consider even more restricted setting which we might call \textit{semi-supervised privileged information}.
Roughly speaking, one could imagine that quantum feature extraction is available for some set of inputs, while labels are available for a (possibly completely non-overlapping) different set. Surprisingly, under contrived distributions (and provably only in this case) learning separations can still be proven even relative to non-uniform learners \ref{sec:crypto_stuff}. This corresponds to the ``DCR /Shadows of QML (1)'' case.
If the input distribution, however, is unique and fixed, an advantage can only occur relative to uniform learners, and even then, we only managed to achieve it assuming highly non-standard, and arguably implausible assumptions. We elaborate on this shortly.

Because of these limitations, this setting is not the main focus of this work, but it is nevertheless interesting that even such partially privileged information can in principl,e offer advantages.

\section{Learning under quantum privileged  information (LUQPI)}
\label{sec:LUQPI}

The LUPI framework, pioneered by Vapnik and Vashist \cite{Vapnik2009}, considers supervised learning where, at training time, each instance \(x\) is accompanied by additional \textit{privileged} information \(x^\star\) that is not available at deployment.
In traditional supervised learning, we only observe pairs datapoint-label \((x,y)\), whereas in LUPI we observe triplets \((x,x^\star,y)\) at training time but must make predictions from \(x\) alone during test and deployment.
Privileged information can, for example, encode explanations, higher-level representations, or information provided by a ``teacher'', and is, intuitively, used to shape more informed decision boundaries and improve generalization.

This paradigm is particularly relevant for our purposes, as it closely matches advantageous offline feature extraction.
In our setting, the privileged information \(x^\star\) is computed from the raw input \(x\) by a quantum feature-extraction procedure.
The learner has access to \((x,E_n(x),y)\) during training but will only see the bare new datapoint \(x\) during deployment. This matches the LUPI pattern exactly, and hence we name it  Learning Under Quantum Privileged  Information (LUQPI)\footnote{In LUPI, there is no explicit specification that the privileged information must be a function of just the current datapoint $x$. One could potentially imagine it also depending on other datapoints, labels, and the specific concept, however, in our case we are interested in this minimal version, where it explicitly is allowed to depend only on the given datapoint, formalized by the notion of the feature map.}.

\subsection{Special classical ML methods for LU(Q)PI}
\label{sec:special-classical}

LUPI and LUQPI settings pose challenges for conventional machine learning pipelines, which typically assume that the same features are available at training and deployment.
In particular, they call for architectures that can exploit augmented training data \((x,E(x),y)\) while still producing predictors that operate only on \(x\) at deployment or test time.

In their original work on LUPI, Vapnik and Vashist introduced SVM+, a modification of support vector machines that incorporates privileged information through a separate correcting function.

In essence, SVM+ uses the privileged data to model the slack variables of a standard soft-margin SVM, which leads to tighter control over the margin and can improve generalization; see rigorous formulations and full details in \ref{SVM+ appendix}. To provide a bit more information at this point, rather than treating slack variables $\xi_i$ as free optimization variables, SVM+ models them as a function of the privileged information: $\xi_i = \langle w^*, \psi(x_i^\star) \rangle + b^*$, where $\psi(\cdot)$ maps privileged features to a (possibly different) feature space. This allows the algorithm to learn from the privileged information which training examples are inherently difficult to classify (requiring large slack) versus those that are easy (requiring small slack), thereby constructing decision boundaries that account for the varying intrinsic complexity of different regions in the input space. Crucially, at test time, SVM+ makes predictions using only the standard features $x$, as the slack variables are no longer needed—the privileged information has served its purpose by guiding the learning of a more informed separator.

More generally, one could imagine architectures that proceed in two stages.
First, a model \(E'\) is trained (possibly implicitly) to predict or approximate the privileged features \(E(x)\) from \(x\).
Then, a second model is trained on triples \((x,E(x),y)\) but is constrained to use \(E'(x)\) at deployment instead of the true privileged features \(E(x)\).
Whenever the privileged features are themselves efficiently learnable from \(x\), such an approach can, in principle, recover the benefits of having access to \(E(x)\) at test time. 

This setting is however restricted, and it is not hard to see it is contained in the semi-supervised case, which has limitations. Note, if the function $E$ is learnable, then it can be learned from any set of examples $(x,E(x)),$ and there is no use for the corresponding label $y$.

The converse need not hold, i.e., it is possible to construct semi-supervised LUQPI settings where the ``fully independent'' learning of $E$ is not possible or required. Examples are scenarios where the learning of $E$ needs to depend on the labeled examples. For instance, the concept class can consist of sub-classes, where the label leaks the subclass, and which can be used to learn a sufficient restriction of $E$ that works on that subclass, but not in general.
In this example, the concepts are tied to distinct distributions (CDS case), but other constructions exist as well. Another example is where a `weaker' feature map $E'(x) = G(E(x))$ for some known 'filter' function\footnote{$G$ can be for example, a simple non-invertable function e.g. a projection on a relevant coordinate of the output of $E$.} $G$ is learnable and suffices for solving the LUQPI task, whereas $E$  itself is not. 
We expect that other, more general constructions could exist showing that semi-supervised advantages can be realized even when the feature map is not learnable by itself. Since these two approaches are then not equivalent, we will refer to the prior as \textit{LUQPI with learnable feature extraction}.
As we explain in the next section, semi-supervised settings are somewhat restricted in the level of classical-quantum separations we can achieve.

Surprisingly, in our separations we will see advantages even in cases where such a two-stage architecture fails because \(E\) is not efficiently learnable, showing that learnability of \(E\) is \textit{not} a necessary condition for there to be a LUPI/LUQPI advantage\footnote{Note, non-learnability of $F$ does not immediately imply that a functioning LUQPI algorithm does not follow the steps of first `trying' to learn $E$, technically failing, but still learning something that is `sufficient'. This will however not be our construction.}.

\section{Provable advantages in LUQPI}
\label{sec:LUQPI-advantages}

In this section, we provide the first main result of our work: the formal constructions that allow for a proof of learning advantages in the LUQPI model.

For didactic purposes, we being by a construction which achieves provable learning advantages in the even more constrained semi-supervised LUQPI setting. This setting has certain unfavorable characteristics, which we also show cannot be avoided in the semi-supervised setting, but then later we show how they can be avoided for true LUQPI.

\subsection{Semi-supervised LUQPI}

The simplest construction for advantages in this setting can be obtained by a careful PAC-like formalization of the ideas from \cite{Jerbi2024}.
For each $n$, we define a family of distributions $\{ {D}_{n}^j\}_j$, where the indexing $j\in 3RSA_n$ is over all 3-RSA integers, so semi-prime integers $ j= p \times q$ ($p,q$ odd) for which $p-1$ and $q-1$ are not divisible by 3. A sample from ${D}^j$ (where we omit $n$ for clarity) is a pair $(x,j)$ where $x$ is a uniformly sampled $n-$bit integer. 

The concept class, for each $n$ can be viewed as a singleton class containing the function: $c(x,j) = DCR_j(x)$ i.e.\ the discrete cube root (DCR) of $x$ modulo $j$.
As explained earlier, we assume that the underlying distribution can be ${D}^j$ for any $j$, and demand that the learning algorithm works for all distributions. 
\paragraph{Classical non-learnability:} The hardness follows standard arguments explored in detail in \cite{modexp} and which we briefly sketch here. We take advantage of the common assumption that the DCR of $x,$ mod $j$ given $x,j$ on input cannot be computed in polynomial time. We note that for a fixed modulus $j$ the discrete cube root is random-self-reducible, implying that the capacity to solve it on average (over $x$) would imply exact solutions in the worst case (over $x$). 
Next, we note that $DCR_j$ has an efficient inverse, so it is random-generatable (meaning \textit{random} input-output pairs can be generated efficiently, see \cite{modexp}) under suitable push-forward distributions, e.g. the uniform distribution.
Furthermore, since we assume the learning algorithm must work \textit{for all distributions and so for all possible moduli} $j$, classical learnability would imply exact worst-case solutions for all $j$ and $x$ without access to data (we can generate it for the desired $j$ efficiently), which is in contradiction with the assumed hardness of DCR.

\paragraph{Quantum learnability:} Since the modulus $j$ is provided in each data point, given access to just one training point, the quantum feature map can output the prime divisors $p$ and $q$ via Shor's factoring algorithm. 
For a fixed 3-RSA modulus $j$, we have that $DCR_j(x) = x^{d_j} \mod \ j$, which is classically efficiently computable given the key $d_j$ which can be efficiently computed from $p$ and $q$.

\paragraph{Semi-supervised learnability:} We note that the distribution itself produces the modulus $j$, which is the only information needed for the feature extractor to provide the factors. Hence, the label is not needed.

\paragraph{Analysis of other (undesirable) properties}
This concept class is clearly pathological as it is a singleton, and this involves no genuine learning. However, this is not a serious issue, as it is possible to combine it with simple learning tasks to make it richer.
For example, one can consider families of concepts $c_k(x,j) \mathop{:=}  DCR_j(x)+ k \mod\ j, $ which leads to the same properties as above and now does involve some learning. 

Another property worth mentioning is that the concept(s) as defined above are not actually classically evaluatable (do not allow polynomially-sized circuits), when taken on the whole domain.
In this sense, the learning separation we obtain is not in identifying the concept but in the hardness of evaluation, which may be undesirable. 
An apparent fix for this is to notice that the class of \textit{restricted} concepts $c_{k,j}(x,j')\mathop{:=} = c_k(x,j),$ where the modulus is fixed for the concept (and indeed the computation is ``correct'' only given the right distribution), is in fact classically tractable as the modulus is now fixed for each concept. 
In this case factors of the modulus can be hard-wired in the polynomially-sized circuit. However, this technically departs from the PAC setting as now we must pair the concepts with the matching distribution, otherwise learning becomes impossible even for a quantum learner as the ``true'' modulus of the concept is unknown. This leads to  what we referred to as the concept-specific-distribution (CDS) setting, and which we do not consider ``canonical'' PAC.

The least desirable property is that in this construction, all the `hardness' is planted in the contrived set of distributions we allow. 
One way of ensuring that the learning hardness comes from the structure of the concept and not the set of distributions is to consider whether it is possible to attain a separation/advantage for a fixed, single (ideally natural), distribution per input size, and not a large set of possible distributions and/or CDS settings.  
However, as we show next for a much more general case of any semi-supervised advantageous setting, having a single distribution is incompatible with having the desired advantage relative to non-uniform classical learners. 

\begin{theorem}
\label{th:semi}
Consider a LUQPI scenario, specified by a concept class $C$ , so which is learnable with an offline quantum feature map, in a semi-supervised setting, and where we fix the input distribution. Then $C$ is learnable with $HeurFBPP/rpoly$ learners\footnote{If the label space is not binary, it is more appropriate to talk about functional classes such as FBPP and not decision classes (e.g. BPP). }. Further, in the case of binary concept classes, it is learnable with $HeurP/poly$ learners. If the concepts are further random-self-reducible, then it is also learnable with $P/poly$ learners.
\end{theorem}
Proof sketch: For the $HeurFBPP/rpoly$ the result is immediate: since there is just one distribution, the randomized poly advice will be the sampled dataset, augmented with whatever the quantum feature map would have computed for each data point (note we do not need the labels, as we assume the class is learnable in the semi-supervised setting, so this is a valid advice/dataset for all possible concepts), note since the advice is the dataset, we note we could have assumed $HeurFBPP/samp$ learners. Furthermore, in the case of binary labels, since the same concept class is in $HeurFBPP/samp$, for binary labels we have that it is in $HeurBPP/samp$ and it holds the latter is contained in $HeurP/poly$ via Adleman's-style amplification arguments \cite{Huang2021, Gyurik_sep}. Random-self-reducibility would finally allow us to conclude it is also learnable with $P/poly$ learners. \qedsymbol.

For the DCR construction as above, one can still try to fix the distribution while maintaining the (weaker) learning advantage relative to uniform classical algorithms. This could be possible if one posits the existence of a sequence $B(n)$ of 3-RSA integers (one per size), for which there exists no poly time algorithm $F(n)$ which returns the factors of $B(n)$. This is not a standard assumption, and it is not clear why such a sequence should exist. Furthermore, to have a reasonable setting for comparison against classical uniform learning algorithms, it would be natural to demand that $B(n)$ itself be a uniform algorithm. This adds to the already demanding assumptions, and it is not clear to the authors whether they could be justified.

Regardless, our objective is to find quantum-offline learning advantages with natural distributions and relative to non-uniform learners.

While, as proven this is impossible in the semi-supervised LUPQI, we next show we can achieve this in the full LUQPI scenario.

We dive in first in the construction itself, clarifying the conceptual parts and possible generalizations later.

\subsection{Main theoretical result}
 In this section, we introduce a concept class that is provably hard to learn by classical non-uniform algorithms under cryptographic hardness assumptions, but easy to learn in the LUQPI setting, relative to a fixed natural (uniform) distribution. We are inspired by cryptographic ideas. At a high level, we define a concept class parameterized by a secret key, where evaluating the concept corresponds to encrypting under this key using a scheme that is classically hard to break, but efficiently learnable by quantum algorithms. For didactic reasons, we will be providing a number of simpler constructions which ultimately fail, but explain the functional role of our rather involved final construction.
 
 The challenging part is to set the scheme up in such a way that quantum preprocessing on the samples alone provides sufficient information to later learn the secret key given the corresponding evaluations. In the following, we assume a cyclic group $\mathcal{G} = (\mathbb{G}, g, q)$ of order $q$, written in multiplicative notation, where group operations and exponentiation are efficiently computable and $g \in \mathbb{G}$ is a generator. For concreteness, we can think of a multiplicative group over integer,s such as a subgroup of $\mathbb{Z}^\ast_p$ for a prime $p$ of the right order. We delve into more details later.

\paragraph{ElGamal encryption.} In our reasoning, we will be building on ideas behind the ElGamal encryption scheme, which we describe briefly for convenience. Given a \textit{public key} $g^y$, a message $m\in \mathbb{G}$ is encrypted as $c:=(g^r, (g^y)^r\cdot m)$, where $r\in\mathbb{Z}_q$ is sampled at random. Here $y\in \mathbb{Z}_q$ is referred to as the \textit{secret key} and is used to decrypt messages. This asymmetry between private and secret keys will be important for our constructions. 

We shall connect this structure to a first attempt at the desired LUQPI concept.
Before stating the concept class, we highlight two properties regarding this scheme.

First property: note that if we define $h:=g^r$ (i.e., define $h$ to be an ``encoding'' of the randomness $r$ in the exponent), we can rewrite the above encryption as $c=(h,h^y \cdot m)$. Hence, given $y$, one can generate an encryption of $m$ only knowing $h$, without knowledge of the exponent $r$ itself. This will ensure our concepts are efficiently computable.
Second property: note that, if only looking at a single ciphertext, the randomness and the secret key are algebraically interchangeable: Given the secret key $y$ and randomness $g^r$, one can decrypt by computing  $m=(g^y)^r\cdot m \cdot (g^r)^{-y}$. Similarly, given the public key $g^y$ and the randomness $r$, one can decrypt by computing $m=(g^y)^r\cdot m \cdot (g^y)^{-r}$. 
This property will be useful toward making the class quantum learnable from features alone.

\paragraph{The concept class - first attempt.} This gives a first attempt for a concept class $C=\{c_{y,m}\colon \mathbb{G}\rightarrow \mathbb{G}^2\mid y\in\mathbb{Z}_p, m\in \mathbb{G}\}$ which is to interpret the input to the concept as randomness $h=g^r$ (note $h$ will be uniformly random as well, corresponding to uniformity of input later in the learning task) and encrypt a message $m$  under the secret key $y$. That is, both the ``secret key'' $y$ and the ``message'' $m$ will specify the concept and thus also enumerate the entire concept class. 

The concepts then output the resulting ``public key'' ($g^y$) together with the second part of the encryption of $m$ as defined by ElGamal (as the first part of the encryption is given via $h$): 
\begin{equation}
c_{y,m}(h)=(g^y,h^y\cdot m).
\end{equation}

This concept has a number of desired properties.
First, toward quantum-offline, it is classically easy to compute each concept given $y$, $m$ (via the first property referred to earlier). Further, it is straightforward to see that if ElGamal encryption is secure (which can be proven secure based on the so-called decisional Diffie-Hellman (DDH) assumption, which is stated in Appendix \ref{sec:crypto_stuff} and loosely related to the average case hardness of the discrete logarithm problem), then this concept class is hard to learn classically. This class is, however, not easily learnable in the LUQPI setting. 

The issue is as follows: Given $h$, a quantum computer allows one to compute $r$ such that $h=g^r$, which, given the corresponding evaluation  $(g^y,h^y\cdot m)$ (i.e., output of the unknown concept indexed by $(y,m)$) allows one to decrypt to classically learn $m$ (see second property above).  It does \emph{not}, however, allow us to learn $y$ given just the feature extraction from the dataset. Therefore, given a fresh $h'$, it is not classically easy to generate $(h')^y\cdot m$ (the output of the concept), as required. 

\paragraph{The concept class - second attempt.} We can patch this by ``leaking'' $y$, more specifically by leaking an encryption of $y$ in the label, which can later be decrypted using feature-extracted information. Recall that $y \in \mathbb{Z}_q$ and that we consider the group $\mathbb{G}$ to be an order-$q$ subgroup of $\mathbb{Z}_q^*$. Fixing the canonical representative of $y$ in $\mathbb{Z}$, the same integer $y$ can be used both as an exponent modulo $q$ and as an element of $\mathbb{Z}_p^*$.

A natural second attempt is to define
\begin{equation}
c_y(h) = (g^y,\, h^y \cdot y),
\end{equation}
where multiplication is performed in $\mathbb{Z}_p^*$. In this construction, the same secret value $y$ is revealed simultaneously in the exponent (via $g^y$) and in the base field (via multiplication by $y$). While this construction can be shown to be learnable in the LUQPI setting, establishing hardness against classical learners is more subtle. While we are not aware of an explicit attack against this construction, it lies outside the scope of standard DDH-based security guarantees. In particular, DDH only protects relations internal to the group, and provides no justification once the discrete logarithm is reused as a group element.

\paragraph{The concept class - informal definition.} This motivates restricting attention to constructions in which all information about $y$ is revealed exclusively through canonical group elements, i.e., elements of the form $g^z$ expressed via the group generator $g$.  In this setting, meaningful evidence can be obtained in idealized models such as the Generic Group Model (GGM), where the adversary is limited to group operations (multiplication, inversion) and equality tests, and cannot directly access exponents. Security in the GGM allows for rigorous hardness statements and provides formal evidence that the assumptions are plausible and resistant to generic attacks.

With this in mind, we alter the concept class as follows. We will now start with a bit string $y=y_1y_2...y_n\in \{0,1\}^n$. Note that we can interpret $y$ as an element in $\mathbb{Z}_q$ as before, by defining $\iota(y)=\sum_{i=1}^{n} y_i \cdot 2^{n-i}$ and taking the result $\mathsf{mod}\; q$.

With this, we can define the concept class 
$C:=\{c_y\colon \mathbb{G}^n\rightarrow\mathbb{G}\times \mathbb{G}^n\mid y\in\{0,1\}^n\}$ as  
\begin{equation}
c_y(h_1,\dots,h_{n})=
\bigl(g^{\iota(y)\;\mathsf{mod}\;q},\{h_i^{\iota(y)\;\mathsf{mod}\;q}\cdot g^{y_i}\}_{i\in[n]}\bigr).
\end{equation}

Hence, we now reveal a bitwise encryption of $y$ in the exponent. Unlike the previous attempt, the security of this construction can be supported under an assumption that has been shown to be hard in the Generic Group Model. We will refer to this class as ElGamal Encrypted Key (EEK) concept class. We will provide a formal definition on more context below.  

\paragraph{The family of groups.}  \label{sec:familyofgroups}

There is still something missing, however, for our purposes. 
So far, the group was treated implicitly and assumed to be fixed but unspecified. To consider hardness, we must instead explicitly specify an infinite family of groups  $\{\mathcal{G}_n\}_{n\in\mathbb{N}}$, where $\mathcal{G}_n=(\mathbb{G}_n,g_n,q_n)$ with $|q_n|=n,$ where $|\cdot |$ here denotes the bitstring-length of the integer $q_n$ written in binary. 

In cryptography, this is typically achieved using \underline{randomness}, i.e., via a probabilistic polynomial-time algorithm $\mathsf{GroupGen}$, which on input $1^n$ outputs a group description $\mathcal{G}_n = (\mathbb{G}_n, g_n, q_n)$ with $|q_n| = n$. 

One way to instantiate $\mathcal{G}_n$ is to sample primes $p_n, q_n$ with $p_n = 2q_n + 1$ and $|q_n| = n$, and let $g_n \in \mathbb{Z}_p^*$ be a random generator of the order-$q_n$ subgroup of $(\mathbb{Z}_{p_n}^*, \cdot)$, which serves as $\mathbb{G}_n$. 
The choice is motivated by the fact that DDH is easy in $\mathbb{Z}_{p_n}^*$ itself, but DDH is conjectured to be hard in the order-$q_n$ subgroup of squares modulo $p_n$.

For our purposes, however, we need the ensemble of groups to be fixed \emph{deterministically} for each $n$. 
While no standard algorithm deterministically outputs a safe prime of bit length $n$ for every $n$, it is plausible to assume that such a deterministic generator exists, assuming safe primes are sufficiently dense among the integers of the desired size. 
Concretely, one can imagine taking a standard randomized safe-prime generation algorithm, which samples candidate primes $q_n$ and tests $p_n = 2q_n + 1$ for primality, and derandomizing it \footnote{
In this context, derandomizing it will simply mean we substitute the random bits with repeated calls to a fixed, deterministic procedure, e.g., a pseudorandom generator, that generates as many bits as needed. While the exact choice of derandomization procedure affects the formal statement of the assumption, for the reasons outlined above, we do not expect it to introduce any meaningful weakness.} to produce a deterministic polynomial-time algorithm $\mathsf{GenSafePrime}(1^n)$ which outputs a safe prime $p_n$ such that $p_n=2q_n+1$ for $|q_n|=n$. 

This gives raise to a deterministic group generator, which we make explicit in the following definition.

    \begin{definition}[Deterministic Group Generation]
\label{def:GroupGen}
Assume that $\mathsf{GenSafePrime}$ is a deterministic polynomial-time algorithm that on input $1^n$ outputs a safe prime $p_n = 2q_n+1$, where $q_n$ is an $n$-bit prime. We define a deterministic polynomial-time group generation algorithm $\mathsf{GroupGen}$ that on input $1^n$ outputs $(\mathbb{G}_n, g_n, q_n)$ as follows: 
Run $p_n \gets \mathsf{GenSafePrime}(1^n)$ and set $q_n = (p_n-1)/2$. Let $\mathbb{G}$ be the order-$q_n$ subgroup of $(\mathbb{Z}_{p_n}^*, \cdot)$. Choose $g_n\in \mathbb{Z}_{p_n}^*$ to be a generator of $\mathbb{G}_n$ in a canonical way (e.g., the smallest integer $g_n > 1$ with order $q_n$ modulo $p_n$). Output $(\mathbb{G}_n, g_n, q_n)$.

\end{definition}

Under the assumption that a deterministic safe prime algorithm exists \footnote{{Such an algorithm will exist if for example it is proven that safe primes occur with sufficient density, and are sufficiently evenly spread out. This is believed to be the case, see e.g. \cite{GopalakrishnaGadiyar2014} }}, the hardness of the Decisional Diffie--Hellman problem and related assumptions in the resulting order-$q_n$ subgroup of $\mathbb{Z}_{p_n}^*$ remains plausible, because security appears to depend on the algebraic structure of the group and not on how the prime modulus $p_n$ was selected. Note that this situation is very different from RSA-like assumptions, where knowing an algorithm that deterministically outputs primes $P$ and $Q$ such that $N = P \cdot Q$ would immediately compromise the hardness of standard problems, such as factoring or the Decisional Composite Residuosity (DCR) problem, as it would reveal the secret factors $P$, $Q$. 
In contrast, for DDH and related assumptions, revealing a deterministic procedure for producing safe primes does not appear to introduce additional algebraic structure that would make Decisional Diffie--Hellman-like problems easier, since hardness depends only on the subgroup structure and not on how the primes were generated.

One might worry that fixing the sequence of primes ahead of time could give rise to non-uniform attacks, where an adversary with preprocessing could prepare information for each group $\mathcal{G}_n = (\mathbb{G}_n, g_n, q_n)$. However, it was shown in \cite{corrigan2018discrete} that in the Generic Group Model, DDH-like assumptions remain plausibly hard even under such preprocessing. More precisely, any generic algorithm for solving the discrete logarithm or decisional Diffie--Hellman problem, even when allowed unlimited preprocessing for each group, cannot perform substantially better than $\Omega(2^{n/3})$. By comparison, without preprocessing, the best generic attack runs in time $O(2^{n/2})$. This result suggests that fixing the group sequence in advance does not meaningfully weaken DDH-like assumptions when analyzed in the Generic Group Model.

For showing hardness in the classical learning setting we will in fact rely on a circular variant of the DDH assumption, which, as mentioned above, has been shown to be hard in the Generic Group Model. 
We note that this assumption has not, to our knowledge, been formally analyzed in the fixed-group setting; nevertheless, we do not expect any asymptotic speed-up beyond what is known for standard DDH, for the same reasons as outlined above.

\paragraph{The distribution.} 

The only remaining piece is to define the input distribution $\mathcal{D}=\{\mathcal{D}_n\}_{n\in \mathbb{N}}$ relative to which we will demonstrate a learning separation. 
To establish hardness against classical algorithms, we take $\mathcal{D}_n$ to be the uniform distribution over $\mathbb{G}_n$, where $\mathbb{G}_n$ is defined as above. 

Note that while this is quite a natural distribution, it does not induce a uniform distribution over the \emph{bitstring representations} of the inputs. In other words, since the distribution has the group as support, it does depend on the concept class (which is defined by the group).  
We will later show, however, how to modify the concept class so as to obtain hardness, under the same assumption, even when inputs are drawn uniformly from $\{0,1\}^{n'}$ for an appropriate choice of $n'$, and thus the distribution is independent from the concept class.

\paragraph{The concept class - formal definition.} 
We are now ready to give the formal definition of the concept class, which we will refer to as \emph{ElGamal Encrypted Key} concept class. To ease readability, we will in the following omit the index $n$,  and write $(\mathbb{G},g,q)$ instead of $(\mathbb{G}_n,g_n,q_n)$. We further write $g^{\iota(y)}$ instead of $g^{\iota(y) \bmod q}$, because $g$ generates a group of order $q$, hence exponentiation in this group implicitly reduces the exponent modulo~$q$.

\begin{definition}[ElGamal Encrypted Key (EEK) Concept Class]\label{def:EEK} Let $\mathsf{GroupGen}$ be a deterministic group generation algorithm (according to Definition~\ref{def:GroupGen}). We define the family of \emph{ElGamal Encrypted Key (EEK) concept classes} $C:=\{C_n\}_{n\in\mathbb{N}}$ relative to $\mathsf{GroupGen}$ as $C_n:=\{c_{y}:\mathbb{G}^n\rightarrow \mathbb{G}\times\mathbb{G}^n\mid y\in \{0,1\}^n\}$, for $(\mathbb{G},g,q)\gets\mathsf{GroupGen}(1^n)$, and
\begin{equation}
c_{y}(h_1,\dots,h_{n})=
\bigl(g^{\iota(y)},\{h_i^{\iota(y)}\cdot g^{y_i}\}_{i\in[n]}\bigr),
\end{equation}
where $\iota(y)=\sum_{i=1}^{n} y_i \cdot 2^{n-i}$. 
    
\end{definition}

The above class is defined in terms of the underlying groups, and as mentioned, the distributions then still explicitly depend on the group and the class. This can be resolved. 
We further give a variant of this concept class, which takes as input bitstrings $\{0,1\}^{n'}$ for a suitable $n'$, and the distribution will be arguably the most natural: uniform over bitstrings, and independent from the concept class.

The idea is to map bistrings to group elements, but this is non-trivial as some mappings will actually jeopardize classical non-learnability.

To exemplify this, consider to take $z_1,\dots,z_n\in \{0,1\}^n$, map $z_i$ into $\iota(z_i)\in\mathbb{Z}$, and define $h_i:=g^{\iota(z_i)}$. 
This is the most natural mapping: bitstring correspond to integers, which exponentiate the generator, ensuring we can cover the group (and nothing but the group).
The resulting concepts would then be of the form 
\begin{equation}
c_y\colon (\{0,1\}^n)^{\times n}\rightarrow \mathbb{G}\times\mathbb{G}^n, (z_1,\dots,z_n)\mapsto \bigl(g^{\iota(y)},\{(g^{\iota(z_i)})^{\iota(y)}\cdot g^{y_i}\}_{i\in[n]}\bigr).
\end{equation}
This concept class, is, however no longer hard to learn classically: An algorithm that knows the input $z_i$ can simply learn $y_i$ by computing 
\begin{equation}
g^{y_i}=(g^{\iota(z_i)})^{\iota(y)}\cdot g^{y_i}\cdot (g^{\iota(y)})^{-\iota(z_i)}.
\end{equation}

The reason why this can be broken is because now the randomness in ElGamal is revealed, which the attacker can easily exploit. 
On a more technical level, a proof of security which would reduce the learning of this new class to the provably secure (EEK) class (which we prove later), would naively involve the inverse of the bistring-to-group-element mapping. However
the issue with this approach is that the map $z_i \mapsto g^{\iota(z_i)}$ is not invertible. 
In fact, invertibility cannot be achieved due to a mismatch in cardinality, but, it is not necessary for the reduction to be achievable. However, there are weaker conditions which suffices for the reduction and which can be achieved.  These conditions are captured by the following notion of a concept-friendly embedding.

\begin{definition}[Concept-friendly embedding]
A mapping 
\begin{equation}
\phi \colon \{0,1\}^{n'} \to \mathbb{G}
\end{equation}
is said to be a \emph{concept-friendly embedding} if it satisfies the following:
\begin{enumerate}
    \item The uniform distribution over $\{0,1\}^{n'}$ induces the uniform distribution over $\mathbb{G}\setminus \{1\}$.
    \item For every $y \in \mathbb{G}\setminus \{1\}$, it is efficient to sample uniformly at random from the preimage $\phi^{-1}(y) \subset \{0,1\}^{n'}$.
    \item The probability that $\phi(z) = 1$ for $z \gets \{0,1\}^{n'}$ is at most $(n \ln n)^{-1}$.
\end{enumerate}
\end{definition}

A remark is in order; for the proofs to go through, properties 1 and 2 are sufficient and exactly needed. However, due to a cardinality mismatch, this cannot be done exactly; to compensate for the mismatch one needs to allow for an additional outcome (the unit of the group). High density of this outcome would invalidate the quantum advantage. While it is impossible to achieve the probability of this event to be zero, we provide constructions which make it (efficiently) decaying with the input size. This will lead to concepts which are hard on everything but a polynomially vanishing fraction of inputs which is sufficient for a separation \footnote{We note that by setting $n' = Omega(n^2)$ we obtain an exponentially vanishing probability of this undesired outcome, but this is overkill. }.

\begin{theorem}\label{thm:conceptfriendlyembedding}
Let $\mathsf{GroupGen}$ be a deterministic group generation algorithm (as in Definition~\ref{def:GroupGen}). Then, there exists $n' \in O(n \log n)$ and a family of explicit mappings
\begin{equation}
\phi = \{\phi_n\}_{n \in \mathbb{N}}, \quad \phi_n \colon \{0,1\}^{n'} \to \mathbb{G},
\end{equation}
which are concept embedding friendly, where $(\mathbb{G}, g, q) \gets \mathsf{GroupGen}(1^n)$.
\end{theorem}

\begin{proof}
Recall that $\mathbb{G}$ is an order-$q$ subgroup of $\mathbb{Z}_p^*$, where $p=2q+1$ is a safe prime. 
To construct $\phi$, we start with a sufficiently long sequence of bitstrings 
\begin{equation}
v = (v_1, \dots, v_\nu), \quad v_i \in \{0,1\}^{\,n+1},
\end{equation}
and select the first \(v_i\) satisfying: 
\begin{itemize}
    \item \(\iota(v_i) \in \{1, \dots, p-1\}\), to ensure a uniform value in \(\mathbb{Z}_p^*\),
    \item \(\iota(v_i) \bmod p \in \mathbb{G}\), to ensure membership in the desired subgroup, which can be efficiently checked.
\end{itemize}

We then define
\begin{equation}
\phi(v) := \iota(v_i) \bmod p.
\end{equation}
If no such \(v_i\) exists within the sequence, we set \(\phi(v) := 1\).

For a single $v_i$, the probability that it does not satisfy either condition is at most $1/4$, plus an additional probability of at most $1/p$ of naturally hitting the identity element in $\mathbb{G}$.  
Hence, the total probability that all $\nu$ attempts fail is at most $(1/4 + 1/p)^\nu$.  
Choosing
$
\nu = \Big\lceil \log_4(n \ln n) \Big\rceil \in \Theta(\log n)
$
makes this failure probability sufficiently small (that is, at most $1/(n \ln n)$). It is further straightforward to verify the first two properties. Consequently, with high probability, $\phi$ satisfies all desired properties and we have $n'=(n+1)\cdot \nu \in O(n\log n)$ as required. 
\end{proof}

\begin{definition}[Binary ElGamal Encrypted Key (BEEK) Concept Class]\label{def:EEK} Let $\mathsf{GroupGen}$ be a deterministic group generation algorithm (according to Definition~\ref{def:GroupGen}). Let $n'\in O(n\log n)$ and $\phi=\{\phi_n\}_{n\in \mathbb{N}}$ as in Theorem~\ref{thm:conceptfriendlyembedding}. We define the family of \emph{Binary ElGamal Encrypted Key (EEK) concept classes} $C:=\{C_n\}_{n\in\mathbb{N}}$ relative to $\mathsf{GroupGen}$ as $C_n:=\{c_{y}:(\{0,1\}^{n'})^n\rightarrow \mathbb{G}\times\mathbb{G}^n\mid y\in \{0,1\}^n\}$, for $(\mathbb{G},g,q)\gets\mathsf{GroupGen}(1^n)$, and
\begin{equation}
c^{\mathsf{bin}}_{y}(z_1,\dots,z_n)
=
\begin{cases}
\bigl(g^{\iota(y)}, \{\phi_n(z_i)^{\iota(y)} \cdot g^{y_i}\}_{i \in [n]}\bigr),
& \text{if } \phi_n(z_i) \neq 1 \text{ for all } i\in [n], \\[6pt]
(1,\dots,1),
& \text{otherwise.}
\end{cases}
\end{equation}

where $\iota(y)=\sum_{i=1}^{n} y_i \cdot 2^{n-i}$. 
    
\end{definition}

It is easy to see that the hardness of learning the binary ElGamal encrypted key concept class implies the hardness of the standard ElGamal encrypted key concept class. We will now show that the reverse is also true. 

\begin{theorem}
Let $\mathsf{GroupGen}$ be a deterministic group generation algorithm (as in Definition~\ref{def:GroupGen}).  If the ElGamal Encrypted Key (EEK) concept class is hard to learn classically under the uniform distribution over $\mathbb{G}^n$ with probability  $1/\mathsf{poly}(n)$ for any polynomial $\mathsf{poly}(n)$, then the Binary ElGamal Encrypted Key (BEEK) concept class is hard to learn under the uniform distribution over $(\{0,1\}^{n'})^n$ with probability better than $1/n + 1/\mathsf{poly}(n)$ for any polynomial $\mathsf{poly}(n)$.
\end{theorem}
\begin{proof} Let $\mathcal{A}$ be an algorithm that learns the BEEK concept class under the uniform distribution with probability at least $1/n+\epsilon_\mathcal{A}$ for some $\epsilon_\mathcal{A}$. We will construct a learner $\mathcal{B}$ for the EEK concept class as follows: Given samples $(h_1^{(j)},\dots,h_n^{(j)},C^{(j)})$, where $C^{(j)}=c_y(h_1^{(j)},\dots,h_n^{(j)})$, $\mathcal{B}$ proceeds as follows:
\begin{itemize}
    \item For each $i\in [n]$, sample $z_i^{(j)}\in \{0,1\}^{n'}$.
    \item If $\phi_n(z_i^{(j)})=1$ for some $i\in[n]$, output $(z_1^{(j)},\dots,z_n^{(j)},(1,\dots,1))$.
    \item Else, for each $i\in [n]$ sample a random $\tilde z_i^{(j)}\in \{0,1\}^{n'}$ in the preimage $\phi_n^{-1}(h_i)$ and output $(\tilde z_1^{(j)},\dots,\tilde z_n^{(j)},C^{(j)})$. 
\end{itemize}
    Since $\phi_n$ is a concept-friendly embedding, it is straightforward to see that this procedure produces random samples for the BEEK concept class, which can serve as input to $\mathcal{A}$. 

    Finally, given a hypothesis $\mathcal{H}$ by $\mathcal{A}$, we can again transform the target sample $(h_1^*,\dots, h_n^*)$ into a target sample $(z_1^*,\dots,z_n^*)$ following the above procedure, and compute $\mathcal{H}(z_1^*,\dots,z_n^*)$. As we have that $\phi_n(z_i^*)=1$ for some $i\in[n]$ with at most probability $1/n$, we must have $\mathcal{H}(z_1^*,\dots,z_n^*)=c_y(h_1^*,\dots, h_n^*)$ with at least property $\epsilon_{\mathcal{A}}$, which proves the claim. 
\end{proof}

Since we proved these two are equivalent in terms of learnabilty, we will give arguments regarding LUQPI learnability and classical non-learnablity for the EEK class and the same will be implied for the BEEK class.

 \newcommand{\GG}{\mathbb{G}}
\newcommand{\Z}{\mathbb{Z}}
\newcommand{\advA}{\mathcal{A}}
\newcommand{\advB}{\mathcal{B}}
\newcommand{\cG}{\mathcal{G}}
\newcommand{\DLOG}{\mathsf{DLOG}}
\paragraph{LUQPI learnability.} 
In this subsection, we will show that there exists a valid LUQPI protocol, allowing quantum feature extraction for the EEK concept class.

\begin{theorem}
    Let $\mathsf{GroupGen}$ be a deterministic group generation algorithm (according to Definition~\ref{def:GroupGen}). Then, the EEK concept class relative to $\mathsf{GroupGen}$ is LUQPI learnable. 
\end{theorem}
\begin{proof}

We define the following family of feature maps, relative to $\GG$:
\begin{equation}
E_n\colon \GG^n\mapsto \mathbb{Z}_q^{n},(h_1,\dots,h_n)\mapsto (\DLOG_{g}(h_1),\dots\DLOG_{g}(h_n)),
\end{equation}
where $\DLOG\colon \GG\rightarrow \mathbb{Z}_q, g^z\mapsto z$ is the discrete logarithm, which is easy to compute quantumly.  

Let $(r_1,\dots,r_n) \gets E_n(h_1,\dots,h_n)$. Recall that, by the properties of the ElGamal encryption scheme, given
\begin{equation}
c_y(h_1,\dots,h_n)
=
\bigl(g^{\iota(y)}, \{h_i^{\iota(y)} \cdot g^{y_i}\}_{i \in [n]}\bigr),
\end{equation}
we can compute, for each $i \in [n]$,
\begin{equation}
g^{y_i}
=
h_i^{\iota(y)} \cdot g^{y_i} \cdot (g^{\iota(y)})^{-r_i}.
\end{equation}
This allows us to recover the secret key bit-by-bit by setting $y_i = 1$ if $g^{y_i} = g$, and $y_i = 0$ otherwise.

\end{proof}

\paragraph{Classical non-learnability.} 

We show that classical hardness follows from a circular version of DDH, which has been proven secure in the Generic Group Model and used in the literature before. For a formal definition, we refer to Appendix \ref{sec:crypto_stuff}.  
\begin{theorem}  Let $\mathsf{GroupGen}$ be a deterministic group generation algorithm (according to Definition~\ref{def:GroupGen}), and 
assume the circular DDH assumption is hard relative $\mathsf{GroupGen}$. Then, the EEK concept class (according to Definiton~\ref{def:EEK}) is hard to learn classically under the uniform distribution. 
\end{theorem} 
\begin{proof} 
First, recall that if the circular DDH assumption is hard relative to $\mathsf{GroupGen}$, then so is the $Q$-times circular DDH assumption, for any polynomial $Q$ (as proven in Appendix \ref{sec:crypto_stuff}. Now, let $\advA$ be a classical probabilistic polynomial time algorithm that learns $c_y$ for all $y$ under the uniform distribution, i.e., given access to $Q$ samples we assume that $\advA$ outputs an efficient hypothesis $\mathcal{H}$ (except with negligible probability), for which $\mathcal{H}(h_1,\dots,h_n)=c_y(h_1,\dots,h_n)$ with probability at least $\epsilon(n)$ for random $
(h_1,\dots,h_n)\in\GG^n$. Then, we define an efficient adversary $\advB$ that has a non-negligible advantage to break the $(Q+1)$-times circular DDH assumption as follows. Let $(g^{\iota(s)},\{g^{b_{i,j}},Z_{i,j}\}_{(i,j)\in[n]\times[Q+1]})$, where $Z_{i,j}=g^{b_{i,j}\cdot \iota(s)}\cdot g^{s_i}$ or $Z_{i,j}=g^{c_{i,j}}$ for random $c_{i,j}\gets\Z_q$. The adversary $\advB$ prepares $Q$ tuples $\{(h_{1}^{(j)},\dots,h_{n}^{(j)}),c_y(h_{1}^{(j)},\dots,h_{n}^{(j)})\}_{j\in[Q]}$ for the learning algorithm $\advA$ as follows. For all $j\in[Q+1]$:
\begin{itemize}
\item $\advB$ defines $h_{i}^{(j)}:=g^{b_{i,j}}$ for all $i\in[n]$. 
\item $\advB$ sets $Y^{(j)}:=(g^{\iota(s)},\{Z_{i,j}\}_{i\in[n]})$. 
\end{itemize} 
Next, $\advB$ runs the learning algorithm $\advA$ on $\{(h_{1}^{(j)},\dots,h_{n}^{(j)}),Y^{(j)}\}_{j\in[Q]}$ and learns hypothesis $\mathcal{H}$. If $\mathcal{H}((h_{1}^{(Q+1)},\dots,h_{n}^{(Q+1)}))=Y^{(Q+1)}$, $\advB$ outputs $1$ (corresponding to a guess that its own input was sampled according to the ``real'' $Q$-times circular DDH distribution, i.e., to the case $Z_{i,j}=g^{b_{i,j}\cdot \iota(s)}\cdot g^{s_i}$), otherwise it outputs $1$. 

Towards the analysis of $\advB$, first note that $\advB$ is efficient. Next, note that if $\advB$ obtained an input with $Z_{i,j}=g^{b_{i,j}\cdot \iota(s)}\cdot g^{s_i}$, then the input provided to $\advA$ is of the form $\{(h_{1}^{(j)},\dots,h_{n}^{(j)}),c_y(h_{1}^{(j)},\dots,h_{n}^{(j)})\}_{j\in[Q]}$ for $y=s$ and uniformly distributed inputs $(h_{1}^{(j)},\dots,h_{n}^{(j)})$. To see this, recall that by construction we have $h_{i}^{(j)}:=g^{b_{i,j}}$.  Further, we have $Z_{i,j}=g^{b_{i,j}\cdot \iota(s)}\cdot g^{s_i}=(h_{i}^{(j)})^{\iota(s)}\cdot g^{s_i}$.  If $\advA$ outputs a hypothesis $\mathcal{H}$ that succeeds with probability at least $\epsilon(n)$, then $\advB$ outputs $1$ with probability at least $\epsilon(n)$ on a \emph{real} circular DDH tuple. On a random circular DDH tuple, on the other hand, $\advB$ outputs $1$ with probability at most $1/q$. Hence, if $\epsilon(n) > 1/\mathsf{poly}(n)$ for some polynomial $\mathsf{poly}$, then $\advB$ is a successful adversary against the circular DDH assumption with advantage at least $1/\mathsf{poly}(n) - 1/q$.

\end{proof}

We summarize the results in the following theorem:  
\begin{theorem}
Let $\mathsf{GroupGen}$ be a deterministic group generation algorithm (according to Definition~\ref{def:GroupGen}), and 
assume the circular DDH assumption is hard relative $\mathsf{GroupGen}$. Then, the PAC learning problem for the EEK concept class exhibits an advantageous quantum-offline LUQPI separation relative to the uniform distribution. Further, assuming the circular DDH assumption is hard relative $\mathsf{GroupGen}$ against non-uniform adversaries, then the separation holds relative to non-uniform classical learning algorithms.  
\end{theorem}

\section{LUQPI in practice}
\label{sec:LUQPI-practice}

The previous sections established that there is room for the existence of advantageous quantum offline feature extraction maps and matched classical algorithms. The constructions were contrived as is usually the case when formal proofs are needed. 
To test the LUQPI framework in more practical conditions, we have investigated a learning problem grounded in quantum physics where we can make informed choices about the physical model, learning task, privileged information, and algorithmic implementation. 
Our goal is not to demonstrate quantum advantage in the complexity-theoretic sense, but rather to investigate whether the LUQPI paradigm can provide practical learning benefits when quantum-derived features serve as privileged information.
This requires a number of choices, namely selecting a quantum many-body model where (i) we can identify a meaningful learning task, (ii) where physical intuition guides what might constitute ``good'' privileged information, (iii) choosing appropriate learning algorithms that can effectively use such information, (iv) and designing experiments that reveal under what conditions the benefits are most pronounced.
In what follows, we first present our choices — the Rydberg atom chain as our physical system, phase classification as the learning task, order parameters as privileged information, and SVM+ as the LUQPI algorithm—and explain the reasoning behind each decision. 

We will begin by addressing the issue of generating training data for quantum systems, and 
then elaborate on what makes privileged information useful in this context, and why we expect the approach to work particularly well under certain training conditions. Following this, we will present the experimental results.

As mentioned, we instantiate our framework by considering a learning problem of ground state phase detection in Rydberg atom chains—a widely studied many-body system that provides an ideal testbed for our approach. The Rydberg system exhibits distinct phases (Disordered, $\mathbb{Z}_2$ ordered, $\mathbb{Z}_3$ ordered) with well-characterized phase transitions. This system offers several practical advantages: the existing physical understanding suggests that order parameters could serve as meaningful quantum privileged information, and numerous datasets for one-dimensional Rydberg chains are readily available \cite{Huang2022, cond_gen_models_github}, enabling reproducible tests of our approach on a physically relevant problem.

In the proposed application, the privileged quantum information consists of expectation values of order parameters computed from ground states. We acknowledge that efficient ground state preparation on quantum computers is itself a challenging problem, which requires special circumstances to be tractable \cite{Cerezo2021}. Despite this limitation, we nonetheless opted for this setting as it allows an easier-to-understand connection to potential LUQPI elements which is central since our focus is the first investigation of LUQPI.

For our numerical experiments, ground states are computed using classical DMRG methods, allowing us to focus specifically on testing the LUQPI framework: whether privileged information in the form of quantum observables (order parameters) can enhance classical learning, independent of how that privileged information is obtained. This separation of concerns—testing the value of privileged information versus the complexity of computing it—allows us to demonstrate the potential of the LUQPI paradigm in a controlled setting. In realistic quantum advantage scenarios, one would require both efficient quantum preparation of approximate ground states and the LUQPI framework to provide learning benefits, but here we isolate and test the latter component.

The question of what constitutes useful privileged information cannot be answered in general. 
Intuitively, privileged information needs to be informative on the problem to be solved, and it needs to provide additional information which is not contained yet in the target quantities, i.e.\ the class labels. 
In the current case, the order parameters reveal information in addition to the phase label, as the value of an order parameter can be viewed as a measure of proximity to the corresponding phase.  

Another important aspect of LUPI and LUQPI is identifying which classical algorithms can effectively use the privileged information. 
As with the privileged information itself, this question cannot be answered in generality but depends on the specific problem and type of privileged information \cite{Vapnik2015, Pechyony2010}. 
Privileged information can improve model performance through various mechanisms, including controlling the uncertainty of model outputs \cite{Lambert2018}, enhancing generalization through more informed regularization \cite{lopezPazUnifyingPI2016}, or improving classification margins by reducing ambiguity \cite{Vapnik2015, sarafianos2017}. In this work, we focus on improvements measured by classification accuracy—the reduction in prediction error on held-out test data.
We chose to use the SVM+ algorithm proposed in the original work on LUPI~\cite{Vapnik2009}. 
Here, the privileged information enters the determination of the slack variables and therefore is most influential close to decision boundaries where the slack variables are relevant. 
This choice also allows direct comparison to plain SVM without privileged information, thereby explicitly elucidating the impact of privileged information.
Additionally, we compare our LUQPI approach against a transformer-based conditional generative model \cite{cond_gen_models}, which represents a state-of-the-art deep learning approach that also operates in an offline setting but learns to predict quantum features explicitly rather than using them as privileged information, as we highlight in the upcoming sections. 

At this point, we foreshadow another experimental design choice having to do with the available sizes of datasets and the importance of data distributions in the success of machine learning. In general, state-of-the-art machine learning models typically require large training datasets to achieve high accuracy~\cite{sun2017revisitin, hestness2017}. For quantum systems, this is particularly problematic: generating training data requires either expensive quantum simulations or real experiments~\cite{Bharti2022, Cerezo2021}, making sample-efficient learning algorithms essential. For this reason, we will be investigating the quality across various but small dataset sizes. Further, standard machine learning practice uses uniformly distributed training samples across parameter space~\cite{settles2009}, but this is inefficient for physical systems where prior knowledge often characterizes behavior over large regions. Ideally, training samples should concentrate where system behavior is unknown and information gain is maximal—typically near phase boundaries or other transition regions. However, such non-uniform sampling introduces data imbalance that can degrade model performance~\cite{HeImbalancedLearningBook2013,krawczykImbalancedReview2016,imbalanceReview2025}. Our experimental design explicitly addresses this trade-off by evaluating multiple sampling strategies that vary in how aggressively training data concentrates near interesting regions.

Having outlined our experimental setup and algorithmic choices, we now introduce the specific quantum physics problem that serves as our testbed: identifying quantum phases of matter in Rydberg atom chains.

\subsection{Quantum Phase Identification: Problem Formulation and Classical Limitations}

Quantum phases of matter represent distinct ground state configurations of quantum many-body systems, each characterized by qualitatively different patterns of expectation values and correlations. 
The boundaries between different phases are marked by phase transitions, which are points in parameter space where the ground state undergoes a qualitative change in its structure. 
In the thermodynamic limit, phases are sharply defined and transitions occur at precise parameter values. However, realistic quantum simulations and experiments operate with finite-size systems where phase boundaries become blurred: quantum correlations extend over length scales comparable to system size, critical fluctuations are pronounced, and order parameters approach their threshold values gradually rather than discontinuously \cite{Sachdev2011}. This makes identifying which phase a finite-size system occupies particularly challenging near phase boundaries, where small changes in Hamiltonian parameters or finite-size effects can lead to ambiguous phase assignments. The task of determining the phase from Hamiltonian parameters without explicitly computing the full ground state constitutes a fundamental challenge in condensed matter physics and quantum simulation \cite{Browaeys2020, Huang2022, cond_gen_models}. The computational expense of generating training data through quantum simulations, combined with the practical challenges of finite-size systems near phase boundaries, motivates the development of data-driven machine learning approaches that can learn accurate phase predictors from limited training samples.

To illustrate these concepts concretely, we consider the one-dimensional Rydberg atom chain. Each Rydberg atom can be either in its ground state $\ket{g}$ (corresponding to qubit state $\ket{0}$) or in the excited Rydberg state $\ket{r}$ (corresponding to qubit state $\ket{1}$). The Hamiltonian for this system of $N$ interacting two-state systems is given by
\begin{equation}
\label{eq:rydberg_ham}
H = \frac{\Omega}{2} \sum_{i=1}^{N} \sigma_i^x - \Delta \sum_{i=1}^{N} n_i + \frac{V_0}{a} \sum_{i<j} \frac{1}{|i-j|^6} n_i n_j
\end{equation}
where $\Omega$ is the Rabi frequency, $\Delta$ is the detuning, $n_i = |r\rangle_i\,{}_i\!\langle r| = \tfrac{1}{2}(\sigma_i^z + 1)$ is the number operator for the atom in the excited Rydberg state at site $i$, $\sigma_i^x$ and  $\sigma_i^z$ are traditional Pauli-X  and Z matrices,  $V_0$ characterizes the van-der-Waals interaction strength between Rydberg atoms that decays as the sixth power of their spatial distance, and  $a$ is the nearest-neighbor distance between atoms. After defining a characteristic interaction length scale as $R_0=(V_0/\Omega)^{1/6}$, the system's behavior can be characterized by two dimensionless parameters: the ratio of interaction range to lattice spacing $(R_0/a)$ and the ratio of detuning to Rabi frequency $(\Delta/\Omega)$.

This system exhibits distinct quantum phases depending on the parameter values. For the one-dimensional Rydberg atom chain, three distinct phases have been identified \cite{Browaeys2020}: a disordered phase and two crystalline phases characterized by periodic spatial ordering patterns of Rydberg excitations. These phases can be characterized by order parameters that detect density-wave ordering with different periodicities. We introduce the quantity \cite{fendley2004} 
 \begin{align}
    \hat O_p=\frac{p}{N}\sum_j^N e^{ik_p j}\, n_j \qquad\text{with } k_p=\frac{2\pi}p\;,\quad p>0
\end{align}
which characterizes density-wave order with period $p$. The notation $\mathbb{Z}_p$ denotes a phase with discrete translational symmetry under shifts by $p$ lattice sites—for instance, $\mathbb{Z}_2$ indicates period-2 ordering (alternating pattern), while $\mathbb{Z}_3$ indicates period-3 ordering (repeating every three sites). While in principle arbitrary periods are possible, the specific interaction form and parameter regime of the Rydberg chain stabilizes primarily $\mathbb{Z}_2$ and $\mathbb{Z}_3$ ordered phases alongside the disordered phase. The disordered phase exhibits no periodic spatial ordering of Rydberg excitations and is characterized by the absence of significant values for all density-wave order parameters $O_{\mathbb{Z}_p}$.

The physical order parameter for period $p$ is calculated as the absolute value of its ground state expectation value,
\begin{align}
\label{eq:order_param}
    O_{\mathbb{Z}_p}=|\braket{\psi_0|\hat O_p| \psi_0}| \equiv |\braket{\hat O_p}|\:,
\end{align}
where $\ket{\psi_0}$ is the ground state of the system. For the $\mathbb{Z}_2$ ordered phase, where atomic states alternate in the pattern ($...rgrgrg...$), the order parameter is
 \begin{align}
\label{eq:order_Z2}
 O_{\mathbb{Z}_2}& = \frac{2}{N} \sum_{j=1}^{N} (-1)^j \braket{n_j} \;,
\end{align}
whereas for the $\mathbb{Z}_3$ ordered phase, where atomic states exhibit the pattern ($...rggrgg...$), the order parameter is
\begin{align}
\label{eq:order_Z3}
 O_{\mathbb{Z}_3}& = \Big| \frac{3}{N} \sum_{j=1}^{N} e^{ik_3 j} \: \braket{n_i} \Big|  .
\end{align}
These order parameters capture the spatial correlations characteristic of each phase. Their values approach unity in the respective ordered phase, while being small (zero in the thermodynamic limit) in other phases.

Following the methodology used in \cite{Huang2022}, phase assignment is performed using the following criteria: 
If $O_{\mathbb{Z}_2} > O_{\mathbb{Z}_3}$ and $O_{\mathbb{Z}_2} > 0.8$, the state is classified as $\mathbb{Z}_2$ ordered. If on the other hand $O_{\mathbb{Z}_3} > O_{\mathbb{Z}_2}$ and $O_{\mathbb{Z}_3} > 0.8$, it is classified as $\mathbb{Z}_3$ ordered; otherwise, when both expectation values are below 0.8, the state is considered disordered. These thresholds reflect the probabilistic nature of quantum measurements and finite-size effects in realistic systems.

Overall, in our setting, phase identification refers to the following learning task: given Hamiltonian parameters $(R_0/a, \Delta/\Omega)$ that specify a particular Rydberg system, predict which of the known phases (Disordered, $\mathbb{Z}_2$, or $\mathbb{Z}_3$) the system occupies, without explicitly computing the ground state or evaluating order parameters at deployment time.

In general, the fundamental computational difficulty of predicting the ground state symmetry and phase of the Rydberg Hamiltonian ~\eqref{eq:rydberg_ham} arises from the long-range interactions between Rydberg atoms along with the exponential scaling of the quantum many-body Hilbert space. 
The van-der-Waals couplings create correlations across the atomic chain that make the ground state properties computationally intractable to determine classically for system sizes beyond approximately 40-50 atoms. Computing the ground state $|\psi_0\rangle$ needed to evaluate  the order parameters $O_{\mathbb{Z}_p}$ 
of Eq.~\eqref{eq:order_param} presents significant challenges. Among prominent classical approaches, exact diagonalization requires constructing and diagonalizing the $2^N \times 2^N$ Hamiltonian matrix, leading to worst-case computational complexity $\mathcal{O}(2^{3N})$ due to the cubic scaling of matrix diagonalization algorithms for general dense matrices. 
Variational approaches such as variational quantum eigensolvers may miss ground state correlations, or fail to find good minima \cite{Cerezo2021}, while tensor network methods become infeasible for highly entangled ground states \cite{Orus2014}. 
Additionally, mean-field approximations that replace the many-body ground state with product states can miss essential quantum correlations distinguishing different phases, particularly near phase boundaries where fluctuations are most significant \cite{Eisert2010}. While these represent some of the most popular classical approaches, the development of efficient methods for quantum many-body systems remains an active and rapidly evolving research area.
 
The straightforward approach to phase detection would require quantum ground state preparation and order parameter measurements for every Hamiltonian instance we wish to classify. 
 
Our data-driven approach instead aims to use quantum devices more strategically: during the training phase, we prepare ground states and measure order parameters for a representative set of Hamiltonian parameters, extracting these quantum features as privileged information. The LUQPI framework then enables a classical model to learn the underlying relationship between Hamiltonian parameters and quantum phases from these training examples. Crucially, once trained, this classical model can predict phases for new Hamiltonian parameters without any quantum measurements during deployment.

At this point, we wish to point out that the problem of identifying the phase (for a reasonable class of systems) in an analogous data-driven approach, but where the phase is predicted from classical shadows of ground states (even when the order parameters are unknown) is fully classically tractable \cite{Huang2022}. In contrast, our task where the inputs are the Hamiltonain parameters, is not believed to be (classically) tractable \footnote{We note that interestingly the same paper (and many follow ups) also prove that the learning of shadows from parameters \textit{within one phase] is also tractable. However the combination is in general not as it crosses phase boundaries. }} 

This offline paradigm becomes especially valuable in data-limited regimes where training samples are scarce and expensive to generate. The approach proposed in the paper is not the only data-driven method for this task, however—alternative approaches such as conditional generative models \cite{cond_gen_models} also make use of quantum-generated data during training but differ fundamentally in their methodology, as we discuss in the next section.  

\subsection{Quantum Feature Extraction vs. Learned Mappings: Methodological Comparison}

Our LUQPI framework differs fundamentally from recent data-driven approaches that attempt to learn explicit mappings from Hamiltonian parameters to quantum observables~\cite{cond_gen_models,fitzekRydbergGPT2024,Huang2021, Huang2022}. These alternative methods, exemplified by the conditional transformer models in \cite{cond_gen_models}, train neural networks to predict ground state expectation values (such as order parameters) directly from Hamiltonian parameters, effectively learning to reproduce quantum features classically. At deployment, these models predict quantum observables without quantum measurements and use these predictions for phase classification, therefore staying in the offline setting.

While both paradigms operate offline, the key distinction lies in how quantum information is leveraged. LUQPI uses quantum-derived order parameters purely as privileged information to guide classical learning, without attempting to learn or reproduce them. The trained model predicts phases directly from Hamiltonian parameters, having benefited from quantum features during training but never needing to compute them. This distinction is significant from a computational complexity perspective. Learning to reproduce quantum observables from Hamiltonian parameters classically requires computing ground state properties, which is intractable in the worst case~\cite{Schuch2009}. While specific problem instances may be tractable, this fundamental computational barrier makes feature learning approaches unlikely to succeed as a general strategy.

It is worth noting here the connection to our theoretical framework developed in the first part of this work. Some alternative approaches, particularly those that learn quantum features from unlabeled data separately from the supervised learning task, effectively operate in what we termed a \textit{semi-supervised privileged information }setting—where quantum feature extraction is available for some inputs while labels are available for a different (possibly overlapping) set. Our LUQPI approach, by contrast, assumes privileged information is available for all labeled training points. As we showed theoretically, genuine LUQPI (with privileged information paired with labels) admits provable strong learning advantages against non-uniform learners under more standard assumptions and natural distributions unlike the semi-supervised variant, where as proven we need to give up on advantages against non-uniform learners under more standard assumptions and natural distributions unlike the semi-supervised variant, where as proven in Theorem \ref{th:semi} we need to give up on advantage against uniformity or on natural distributions. Our experiments investigate whether this full LUQPI paradigm provides measurable practical advantages over both classical baselines and methods that explicitly learn quantum feature mappings.  

\subsection{Experimental Results and Performance Analysis}

The purpose of this experimental study is to investigate whether the LUQPI paradigm provides detectable learning advantages when quantum-derived features serve as privileged information in a small case study. Specifically, we test the hypothesis that quantum order parameters, when used as such features, enable better phase classification performance compared to purely classical learning. We also investigate under what conditions (data availability, sampling strategies) these advantages are most pronounced. We acknowledge that the order parameters used here require ground state computation, which is not efficiently computable on quantum hardware in general— as explained earlier, we chose this approach for its strong connection to our theoretical framework and to isolate the LUQPI mechanism from feature computation complexity. Future work should explore privileged information based on more realistic assumptions.

We evaluated the proposed LUQPI approach on a ground state phase detection task for a one-dimensional chain of Rydberg atoms with 31 atoms across different training dataset sizes. We use the dataset with 1152 samples proposed and used in Ref.~\cite{cond_gen_models,cond_gen_models_github}

which consists of Hamiltonian parameters $(\Delta/\Omega, R_0/a)$ spanning the phase diagram (see top left panel of Figure~\ref{fig:misclassification} below).
The ground states were  computed using the density matrix renormalization group (DMRG)  and the order parameters $O_{\mathbb{Z}_2}$ and $O_{\mathbb{Z}_3}$ are evaluated from these ground states. For a learning algorithm, inputs are Hamiltonian parameters and the prediction target is the phase label: disordered, $\mathbb{Z}_2$-ordered, or $\mathbb{Z}_3$-ordered.

We implement LUQPI using SVM+ as our primary algorithm, which incorporates quantum features through the slack variable mechanism as described in ~\ref{sec:lupi}. For comprehensive comparison, we evaluate against: (1) a classical SVM baseline (no privileged information), and (2) the transformer-based conditional generative model from~\cite{cond_gen_models} that learns to predict quantum features explicitly rather than using them as privileged information (as analyzed in the previous section). The classical SVM learns this mapping directly from training data. SVM+ receives the same inputs but additionally accesses order parameters as privileged information during training only; at deployment, both SVMs use only Hamiltonian parameters. Since SVM is inherently binary, we employ one-versus-all multi-class classification with three binary classifiers. The transformer model learns to predict local POVM expectation values from Hamiltonian parameters, then uses these predicted features for phase classification.

To address the realistic constraint that training data is expensive, we evaluate three sampling strategies reflecting different data collection priorities: \textit{uniform sampling} across parameter space, \textit{light boundary sampling} with moderate concentration near phase boundaries, and \textit{hard boundary sampling} with strong concentration at boundaries. Models were trained on varying sample sizes from 15 to 100 points. All strategies maintain class balance proportional to the dataset (disordered:$\mathbb{Z}_2$:$\mathbb{Z}_3$ ratio of 56:27:17), which represents a sufficiently weak imbalance that specialized techniques are unnecessary~\cite{HeImbalancedLearningBook2013}.

Since SVM's are sensitive to the values of their hyperparameters \cite{chapelle2002choosing}, we performed a hyperparameter selection procedure.
For each of the three sampling methods, we selected the dataset with $N_{train}=40$ training dataset points. 
We trained models with different hyperparameters and selected the best performing parameters, which are then used for all other experiments for the same sampling strategy.
The performance of each hyperparameter configuration was estimated with $5$-fold cross-validation using average validation error as a performance metric. 
As hyperparameters we tested radial-basis function (RBF) and polynomial kernels for the SVM models, and for each kernel, we did a grid-search of the most relevant hyperparameters, which are for the SVM
the regularization parameter $C=\{1,10,50,100,250,500,1000, 2500\}$ and the kernel parameter $\gamma=\{0.01,0.1,1,10,100\}$. 
For the SVM+, we tested the same parameter values for the standard SVM part, and used an RBF kernel with  
$\gamma^*=\{10^{-5}, 10^{-4},10^{-3},0.01,0.1,1\}$  for the slack-variable part. The search values for the regularization parameter of SVM+  we used  $C^*=\{1, 10, 100, 10^{3}, 10^{4}, 10^{5}\}$ (see Sect.~\ref{eq:svmp_primal} for more details on these parameters).
Note that this is not a proper hyperparameter optimization in the common sense, since we did not tune the precise parameter values to optimal performance, but rather only determined their rough order of magnitude. 
Therefore, we expect our results from SVM and SVM+ to be lower bounds on the achievable performance, and expect that with more fine-tuning, better results could be achieved. 

For the transformer model, we did not employ any hyper-parameter tuning due to the large demand for computational resources required to train the models, and instead used the parameters of the original work~\cite{cond_gen_models,cond_gen_models_github}. 
Moreover, the type of transformer architecture used in the original paper was shown to be rather robust to changes of certain hyperparameters~\cite {xiongTransformer2020}. 

All the final performances were evaluated by computing classification accuracy on the whole dataset, uniformly distributed across the parameter space. 
Each experiment was repeated with 30 different randomized training and test datasets to obtain statistically robust results with 95\% confidence intervals.

\begin{figure}[!htbp]
    \centering
    \begin{subfigure}{0.75\linewidth}  
        \centering
        \includegraphics[width=\linewidth]{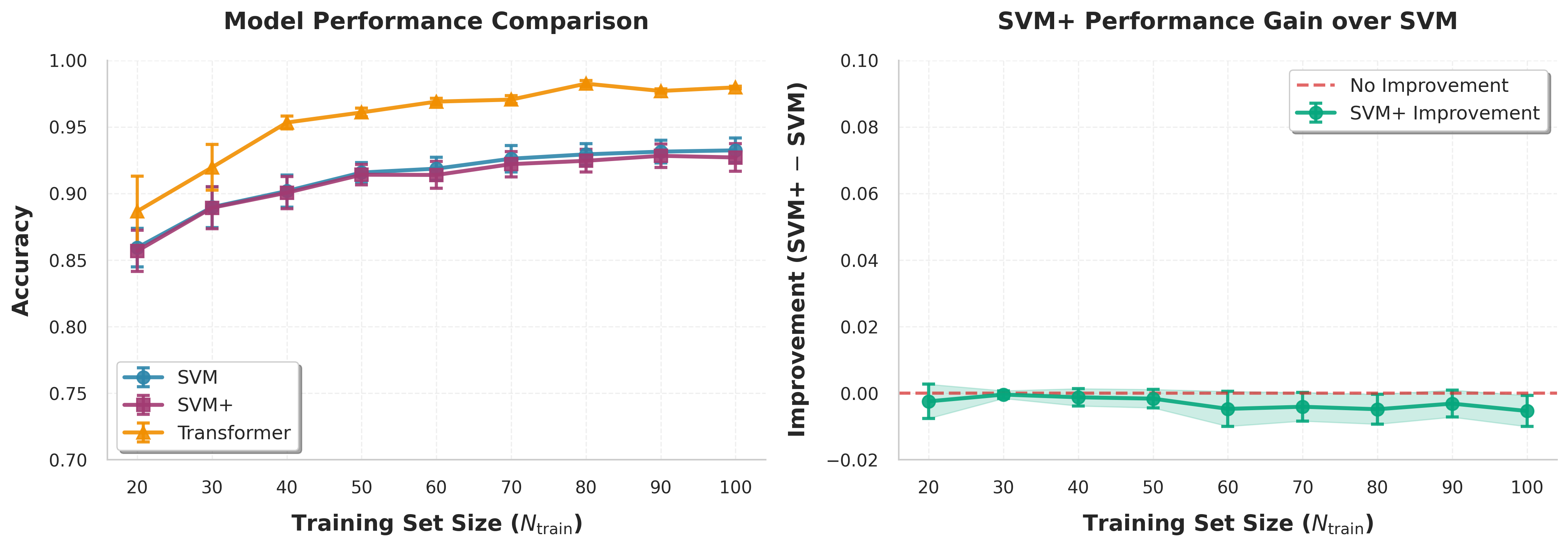}
        \caption{Uniform sampling: modest consistent gains.}
        \label{fig:subfig_uniform}
    \end{subfigure}

    
    \begin{subfigure}{0.75\linewidth}
        \centering
        \includegraphics[width=\linewidth]{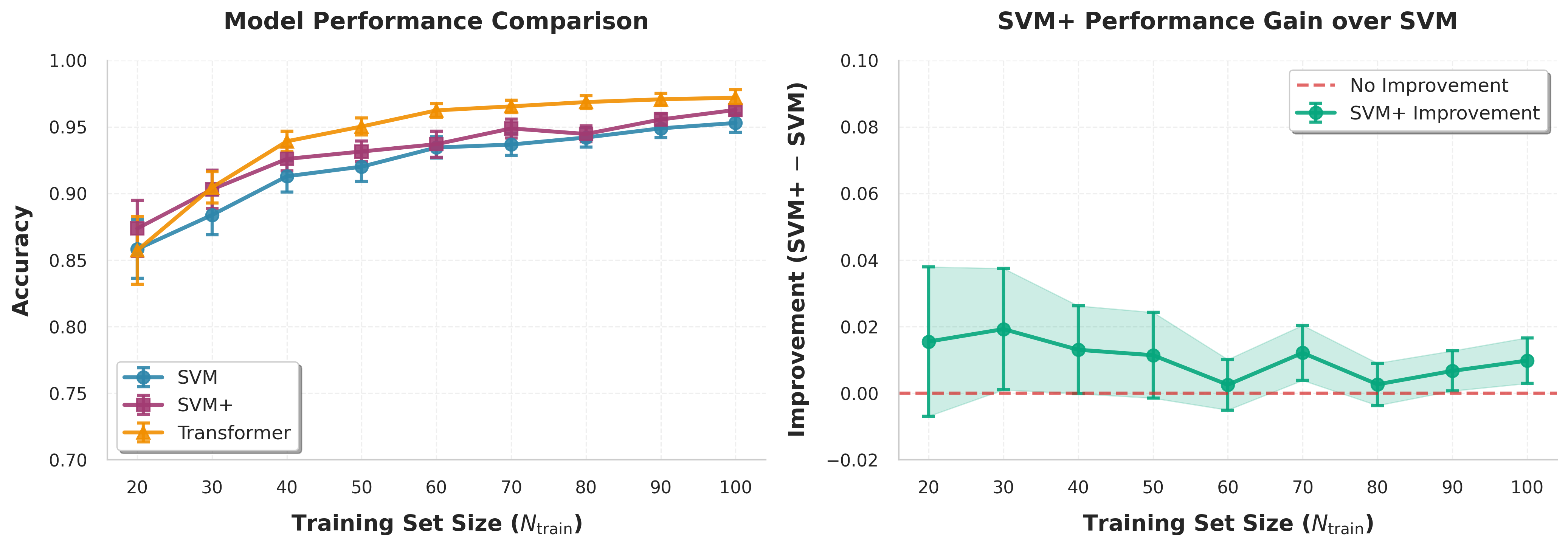}
        \caption{Light boundary sampling: moderate gains decreasing with training size.}
        \label{fig:subfig_lightbd}
    \end{subfigure}
    
    \begin{subfigure}{0.95\linewidth}
        \centering
        \includegraphics[width=0.75\linewidth]{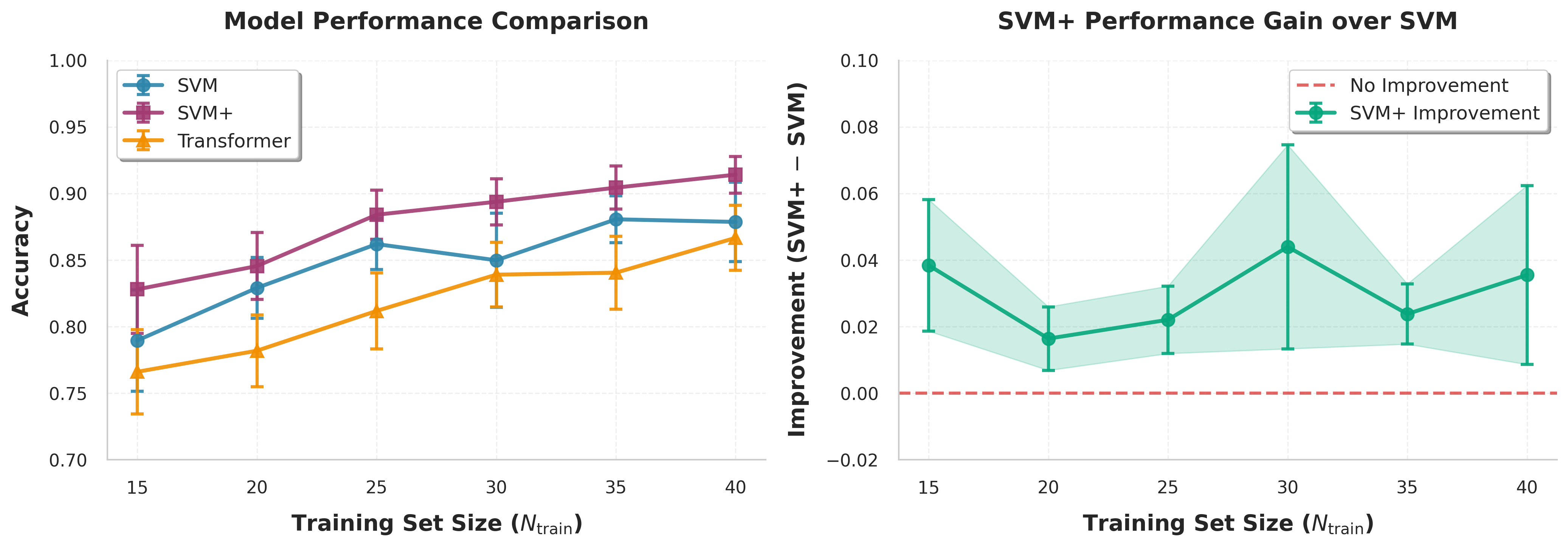}
        \caption{Hard boundary sampling: substantial gains in low-data regime.}
        \label{fig:hard_boundary}
    \end{subfigure}

    \caption{Model performance comparison across sampling strategies. Left panels show test accuracy for SVM (blue), SVM+ (purple), and Transformer (orange) as functions of training set size. Right panels show SVM+ performance gain over SVM with 95\% confidence intervals (green shading indicates improvement).}
    \label{fig:combined_comparison}
\end{figure}

Figure~\ref{fig:combined_comparison} presents the performance (accuracy of correct phase predictions) as a function of training dataset sizes of the three methods, averaged over 30 different randomized training datasets for all three sampling strategies. 
Under uniform sampling (Figure~\ref{fig:subfig_uniform}), the transformer model consistently achieves the best performance throughout all training sizes, while SVM and SVM+ showing comparable performance.
All methods show a tendency to reduce performance when trained on a smaller number of samples, but the decrease is rather moderate. 
The transformer model goes from about 98\% with 100 training samples to still rather large 88\% accuracy for only 20 training samples. 
SVM and SVM+ degrade from 93\% for 100 samples to 86\% accuracy for 20 samples.       

For light boundary sampling (Figure~\ref{fig:subfig_lightbd}), the transformer continues to perform best overall, but both SVM-based models become comparable especially for a small number of training samples. 
The SVM+ model shows an improvement over SVM of approximately 1-2 percent for training set sizes up to 50, while being comparable for larger training sets. 
As in the uniform sampling case, the performance decrease with smaller training set sizes is only moderate.
Additionally, all models handle the non-uniform sampling in the training dataset since the performance values are comparable to the uniform case of Fig.~\ref{fig:subfig_uniform}. 

In the hard boundary sampling shown in Fig.~\ref{fig:hard_boundary} SVM+ shows overall best performance, with an 2-4\% mean performance increase over SVM throughout the low-data regime (15-40 samples). 
It appears that classical SVM is comparable and even outperforms the transformer-based model in the low-data regime.
The overall performance reduction of all models compared to light boundary or uniform sampling is more pronounced as it drops about 5-10 percentage points.    

\begin{figure}[!htbp]
    \centering
    \includegraphics[width=0.75\linewidth]{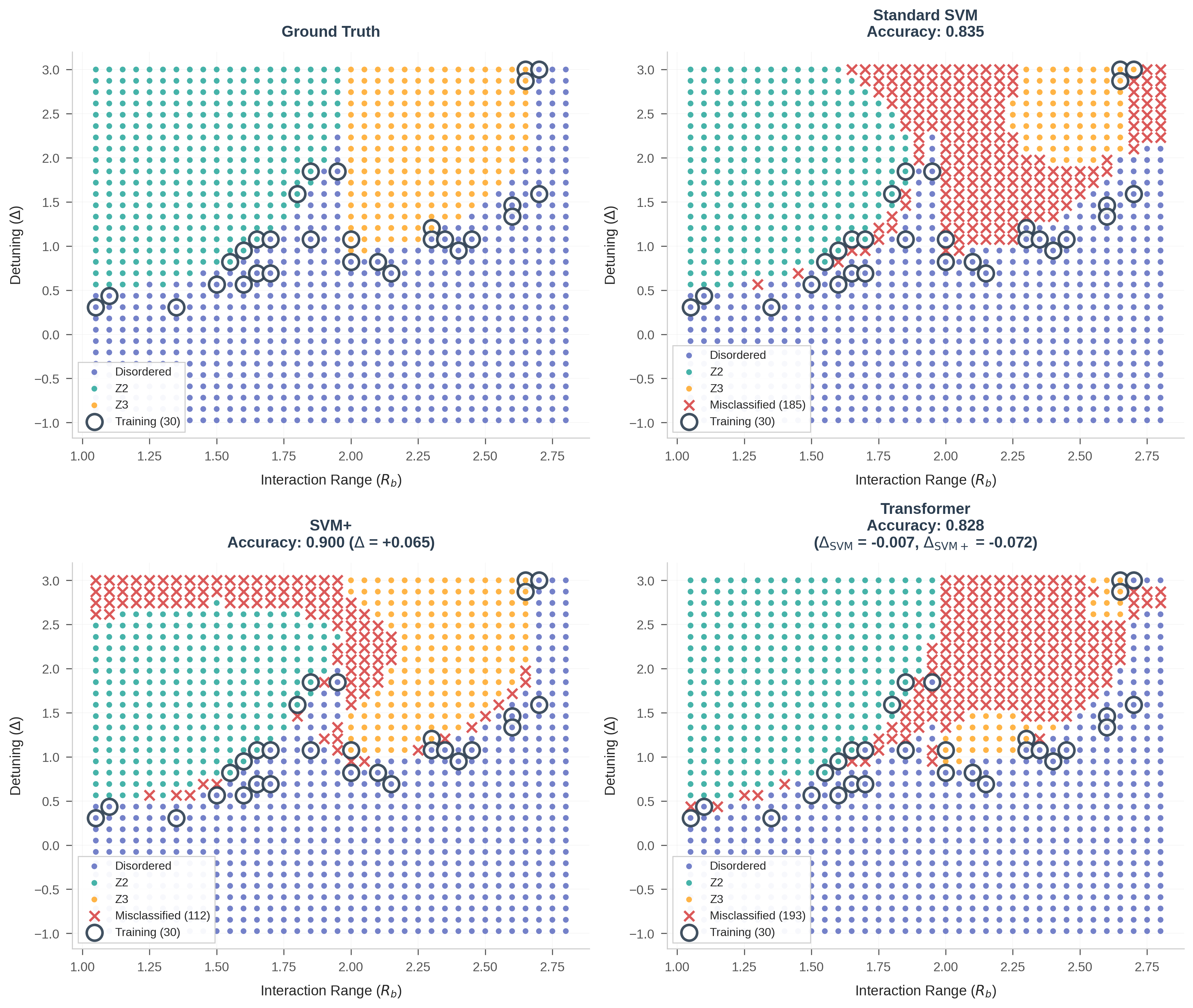}
    \caption{Misclassification patterns in hard boundary sampling with 30 training samples (black circles). Ground truth (top left) compared to Standard SVM (top right), SVM+ (bottom left), and Transformer (bottom right). Red X marks indicate misclassified test points.}
    \label{fig:misclassification}
\end{figure}

To further elucidate the performance of the models, we analyze the distribution of misclassifications in the input parameter space. 
Figure~\ref{fig:misclassification} shows a representative example with 30 training data samples for the case of hard boundary sampling distribution. 
All models show extended regions of misclassified points.
Some of those areas indicate the generally difficult task for any machine learning model to extrapolate/generalize to regions with no training data at all,  like the top middle to left part of the phase diagram in the example of Fig.~\ref{fig:misclassification}. 
Other misclassifications occur in regions which are close to the decision boundary and where training samples are close by. 
These are the areas where privileged information is beneficial and allows for more reliable predictions. 
In the shown example, the wrong predictions of the standard SVM in the top right region are correctly predicted by the SVM+. 
Similarly, the center region where all three phases are close, the SVM+ can predict more reliably than the transformer model.

The first conclusion to be drawn from these results is that all three methods perform well in the low data regime when trained with a small number of samples. 
Second, all models can be successfully trained with samples from highly non-uniform distributions and still generalize to the complete phase space, i.e.\ predict phase labels reliably in regions where no training data was located.

Regarding the privileged information, our results suggest that it is most useful close to phase boundaries. 
There is no substantial improvement of SVM+ over SVM in the case of uniform training samples. The majority of training samples are far from the boundaries where order parameters do not provide additional information.
In contrast, for boundary-biased distributions with increased density of training samples in the vicinity of phase boundaries, the additional information becomes more valuable. The algorithm uses it to distinguish more reliably between training examples that are inherently difficult to classify (near true boundaries) versus those that are mislabeled or noisy. 
In low-data regimes, some phases can be significantly underrepresented in the training set, leading classical SVM to effectively ignore samples from these minority classes. The privileged information in the form of the  order parameters helps SVM+ maintain classification performance even for underrepresented phases by providing direct insight into the phase structure. 
This explains why the accuracy gap between SVM and SVM+ is larger in boundary sampling scenarios (light and hard) compared to uniform sampling: boundary regions contain more ambiguous cases where privileged information is most valuable, and the combination of low data and challenging classification amplifies the benefit of quantum features.

\section{Discussion}

This work introduces the Learning Under Quantum Privileged Information (LUQPI) framework, a systematic approach for integrating quantum-derived features into classical machine learning. The central contribution is demonstrating that quantum computers can provide provable learning advantages even when used minimally—solely during training to extract features that guide classical learning, with no quantum resources required during deployment.

We established both theoretical foundations and investigated a practical application of this paradigm. Theoretically, we introduced the notion of advantageous feature extraction and constructed learning problems where quantum features transform classically intractable tasks into efficiently learnable ones, providing formal proof of quantum advantage within the LUQPI setting. Our classification framework, based on deployment strategies, label usage, and data dependencies, shows that various existing quantum machine learning approaches can be understood as different instantiations of quantum feature extraction, with the conceptually minimal usage—offline, label-independent, data-dependent extraction—naturally connecting to the established Learning Under Privileged Information (LUPI) framework.

To test these theoretical insights in a more realistic setting, we conducted a systematic investigation of quantum phase identification in Rydberg atom chains—a problem particularly well-suited to the LUQPI framework. This application demonstrates how quantum-derived order parameters, serving as privileged information during training, enable classical models to predict phases from Hamiltonian parameters without quantum measurements at deployment, which is what the offline paradigm of our framework prescribes.

In our experiments, motivated by the practical constraint that quantum data is resource-intensive to generate, we focused on small training sets (15-100 samples) and considered both uniform and non-uniform parameter space distributions. We implemented LUQPI using the SVM+ learning model with order parameters as privileged information, and compared against classical SVM and a transformer-based conditional generative model that learns to predict POVM measurements from Hamiltonian parameters and uses these predictions for phase classification. Our experiments reveal that quantum privileged information can provide consistent accuracy improvements, with advantages most pronounced when two conditions coincide: boundary-concentrated sampling and severely limited training data. In the hard boundary regime with 15-40 samples, SVM+ outperforms both classical SVM and transformer baselines. 

Overall, the LUQPI framework, supported by our theoretical proofs and numerical demonstrations, addresses practical quantum machine learning challenges by advocating minimal quantum resource usage—generating features during training to improve classical models while maintaining efficient classical deployment. 

\section{Limitations and future directions}

While we have formally proven the possibility of exponential advantages in LUQPI scenarios for (heavily) contrived settings and provided initial evidence for LUQPI's practical utility, we explicitly acknowledge several limitations of our empirical study, possible mitigations and thereby motivated lines of future work:
(i) We studied a single system (1D Rydberg phase prediction) that is fully classically tractable. This was nonetheless a useful first step precisely because all quantities can be computed exactly and compared reliably.
(ii) We used order parameters as privileged information, which requires ground state computation—a generally intractable task, as discussed earlier. 
(iii) Related to (ii), the privileged information consisted of expectation values of order parameters, which are typically not known a priori for novel quantum systems and would need to be identified or learned in realistic applications.
(iv) All quantities (including quantum features and labels) were computed with numerical precision and essentially zero error. In practice, quantum feature extraction on real devices would incur measurement noise and systematic errors, especially on near-term quantum hardware.

These motivate future research directions as follows. Regarding (i), LUQPI will be of greatest interest for classically intractable systems, but which can still be accessed with near-term quantum devices, for example, quantum simulators. Of immediate significant interest would be 2D versions of common systems. However, in these cases, it will be a challenge to collect the data at more substantial sizes and to verify the outcomes. Indications can be obtained by sacrificing large off-line compute time to obtain the relevant information, while the LUQPI machinery is kept very lightweight.
On the subject of (ii), using order parameters was theoretically motivated (see introduction in section \ref{sec:LUQPI-practice}), but there are no reasons to believe this was optimal for LUQPI needs. We foresee fruitful investigations in this direction, studying the utility of a number of viable candidates for privileged information (for the phase identification task), from using time-evolution information (motivated by response function theory), noisy Gibbs states which can be prepared in labs, to so called trapped states (see \cite{phasecraft}), all of which are more tractable quantumly.
Concerning (iii) - the fact that order parameters are known in advance - we conjecture this could be circumvented in some cases, for example, if the privileged information were the shadows of ground states themselves. 
In particular in \cite{Huang2022}, it is shown that the order parameters can be (provably) learned from shadows, which could be combined with our LUQPI method,  although whether this would work is not trivial \footnote{In more detail, the learning of order parameters as described in \cite{Huang2022} assumes access to the order parameter expectations $\langle O \rangle$. There the mapping from shadows to the expectation is  learned using Lasso-regression-related methods. In our setting we can however only assume access to the phase label, which are the thresholded version of the expectation, roughly $sign(\langle O \rangle+b)$, but not the order parameter expectation itself. Given the known relationships between 0-1 losses (here corresponding to the thresholded values) and convex surrogates (actual expectations) (see
[https://statistics.berkeley.edu/sites/default/files/tech-reports/638.pdf)]) we conjecture sufficiently good approximations of order parameters could be learned.}. In relation to (iv), we used numerically exact, noiseless simulations rather than realistic experimental conditions with measurement shot noise and approximate ground states. We expect LUQPI to demonstrate even greater advantages in noisy regimes where privileged information disambiguates samples and enables more robust models. More fundamentally, realistic quantum advantage scenarios require efficient approximate ground state preparation on quantum devices, not exact classical simulation. Investigating LUQPI performance with thermal Gibbs states or dynamical quantum features from pump-probe experiments~\cite{phasecraft} represents an important direction for future work, bridging the gap between our theoretical framework and near-term experimental capabilities. Regarding the precision, indeed further studies investigating the robustness of the methods to noise and errors are warranted. SVMs are somewhat robust under small label error, however, it is unknown to what extent SVM+ is robust to privileged information error.

We also highlight other research lines of interest not related to the listed study limitations. In general, the question of what constitutes useful privileged information is of key importance. Further, very few classical models are well-suited for the use of privileged information (i.e., in the quantum-offline mode) and we expect much progress can be made here. Lastly, there is nothing specific to LUQPI which makes it particularly well suited to the phase detection problem, and it is interesting to investigate it for the prediction of other quantities, such as dynamical variables, out-of-time-ordered correlators, in essence, any quantity where we suspect quantum computers may offer substantial advantages.

In summary, the proposed LUQPI framework was shown to be beneficial in scenarios where quantum measurements provide information that is both physically relevant and difficult to extract from classical features alone. 
We expect our approach to be valuable for experimental quantum simulation studies, where generating extensive training datasets is experimentally challenging and computationally expensive.

\section{Acknowledgments}

The authors thank Adri\'{a}n P\'{e}rez-Salinas, Hao Wang, and Anastasiia Skurativska for valuable discussions.
Vasily Bokov acknowledges funding from the Honda Research Institute Europe.  The work of Lisa Kohl is funded by NWO Talent Programme Veni (VI.Veni.222.348) and by NWO Gravitation project QSC.
Part of this work was also supported by the Dutch National Growth Fund (NGF), as part of the Quantum Delta NL programme. This work was also supported by the European Union’s Horizon Europe program through the ERC CoG BeMAIQuantum (Grant No. 101124342).
\newpage

\appendix

\section{LUPI}
\subsection{Learning Under Privileged Information (LUPI) Framework}
\label{sec:lupi}

The LUPI framework, pioneered by Vapnik and Vashist \cite{Vapnik2009}, addresses scenarios where additional information is available during training but not during testing. This "privileged information" can guide the learning process to achieve better generalization performance, even though it cannot be used for actual predictions. In our quantum-enhanced learning setting, the privileged information naturally corresponds to quantum features that may be expensive or impractical to compute during deployment but can be extracted during the training phase.

In traditional supervised learning, we have access to training pairs \((\mathbf{x}_i, y_i)\) where \(\mathbf{x}_i\) represents the input features and \(y_i\) the corresponding labels. The LUPI framework extends this by introducing an additional information source \(\mathbf{x}^*_i\) available only during training, resulting in triplets \((\mathbf{x}_i, \mathbf{x}^*_i, y_i)\). The key insight is that this privileged information \(\mathbf{x}^*_i\) can inform the learning algorithm about the difficulty or structure of individual training examples, leading to more informed decision boundaries and improved generalization.

This paradigm is particularly relevant for quantum machine learning because quantum devices can often extract features that provide deep insights into the underlying structure of data, but may be too resource-intensive to compute during routine deployment. By treating quantum-derived features as privileged information, we can harness their power during training while maintaining practical deployment constraints.

\subsection{Support Vector Machine with Privileged Information (SVM+)}
\label{SVM+ appendix}

To illustrate the LUPI framework concretely, we revisit the Support Vector Machine with Privileged Information (SVM+) \cite{Vapnik2009, Vapnik2015}, which serves as our testbed algorithmic tool for leveraging quantum features in classical learning. The SVM+ model extends the classical Support Vector Machine (SVM) framework to incorporate \textit{privileged information}---additional information about training instances that is available only during training but not at test time. This paradigm reflects how human teachers provide students with explanatory information, analogies, or contextual insights during instruction that help accelerate learning, even though these explanations are not available during the final examination \cite{Vapnik2009}.

We focus on the binary classification setting, where the goal is to learn a decision boundary that separates data points into two classes based on their labels $y_k \in \{-1, 1\}$. The training data consists of triplets $(x_k, x^*_k, y_k)$ for $k = 1, \ldots, N$, where $x_k \in \mathcal{X}$ represents the standard input features (available at both training and test time), $x^*_k \in \mathcal{X}^*$ represents the privileged features (available only during training), and $y_k$ is the class label.

In a standard soft-margin SVM, we solve the following primal optimization problem:
\begin{align}
\min_{\mathbf{w}, b, \boldsymbol{\xi}} \ & \frac{1}{2} \|\mathbf{w}\|^2 + C \sum_{k=1}^{N} \xi_k, \label{eq:svm_primal}\\
\text{subject to } \ & y_k(\langle \mathbf{w}, \phi(\mathbf{x}_k) \rangle + b) \geq 1 - \xi_k, \quad \xi_k \geq 0, \quad k = 1, \ldots, N,
\end{align}
where $\phi: \mathcal{X} \to \mathcal{Z}$ is a feature map transforming inputs into a higher-dimensional feature space $\mathcal{Z}$, $\mathbf{w} \in \mathcal{Z}$ defines the normal vector to the separating hyperplane, $b \in \mathbb{R}$ is the bias term, $\xi_k \geq 0$ are slack variables that allow for misclassification or margin violations, and $C > 0$ is the regularization parameter controlling the trade-off between maximizing the margin and minimizing classification errors.

The decision function for a new input $\mathbf{x}$ is given by $f(\mathbf{x}) = \text{sign}(\langle \mathbf{w}, \phi(\mathbf{x}) \rangle + b)$. In this formulation, the algorithm must estimate $n$ parameters for the weight vector $\mathbf{w}$ (where $n = \dim(\mathcal{Z})$) plus $N$ slack variables $\xi_k$, totaling $n + N$ parameters. The convergence rate for achieving a specific accuracy scales as $O(\sqrt{h/N})$, where $h$ is the VC dimension of the hypothesis space and $N$ is the number of training samples \cite{Vapnik2009}.

In contrast, SVM+ introduces a second function $\psi: \mathcal{X}^* \to \mathcal{Z}^*$, defined on the privileged information space $\mathcal{X}^*$, to explicitly model the slack variables. Rather than treating slack variables as independent optimization variables, SVM+ parameterizes them as:
\begin{equation}
\xi_k = \langle \mathbf{w}^*, \psi(\mathbf{x}^*_k) \rangle + b^*, \label{eq:slack_model}
\end{equation}
where $\mathbf{w}^* \in \mathcal{Z}^*$ and $b^* \in \mathbb{R}$ define a \textit{correcting function} in the privileged space. This functional form allows the algorithm to learn which training examples are inherently difficult to classify based on patterns in the privileged features. The SVM+ primal optimization problem becomes:
\begin{align}
\min_{\mathbf{w}, \mathbf{w}^*, b, b^*} \ & \frac{1}{2} \|\mathbf{w}\|^2 + \frac{C^*}{2} \|\mathbf{w}^*\|^2 + C \sum_{k=1}^{N} \left[\langle \mathbf{w}^*, \psi(\mathbf{x}^*_k) \rangle + b^*\right], \label{eq:svmp_primal}\\
\text{subject to } \ & y_k(\langle \mathbf{w}, \phi(\mathbf{x}_k) \rangle + b) \geq 1 - \left[\langle \mathbf{w}^*, \psi(\mathbf{x}^*_k) \rangle + b^*\right], \label{eq:svmp_constraint1}\\
& \langle \mathbf{w}^*, \psi(\mathbf{x}^*_k) \rangle + b^* \geq 0, \quad k = 1, \ldots, N, \label{eq:svmp_constraint2}
\end{align}
where $C^* > 0$ is a regularization parameter for the correcting function in the privileged space. The constraint (\ref{eq:svmp_constraint2}) ensures that the modeled slack variables remain non-negative.

The fundamental difference between SVM and SVM+ lies in how they handle the trade-off between margin maximization and training error:

\begin{enumerate}

\item \textit{Slack variable modeling}: In SVM, each slack variable $\xi_k$ is independent and directly penalized in the objective. In SVM+, the slack variables are modeled as outputs of a function $\langle \mathbf{w}^*, \psi(\mathbf{x}^*_k) \rangle + b^*$, which introduces structural dependencies between different training examples based on their similarity in the privileged space.

\item \textit{Dual formulation}: The dual formulation reveals the deeper connection. For standard SVM, the dual is:
\begin{align}
\max_{\boldsymbol{\alpha}} \ & \sum_{k=1}^{N} \alpha_k - \frac{1}{2} \sum_{k,k'=1}^{N} \alpha_k \alpha_{k'} y_k y_{k'} K(\mathbf{x}_k, \mathbf{x}_{k'}), \label{eq:svm_dual}\\
\text{subject to } \ & \sum_{k=1}^{N} \alpha_k y_k = 0, \quad 0 \leq \alpha_k \leq C,
\end{align}
where $K(\mathbf{x}_k, \mathbf{x}_{k'}) = \langle \phi(\mathbf{x}_k), \phi(\mathbf{x}_{k'}) \rangle$ is the kernel function in the decision space. For SVM+, introducing Lagrange multipliers $\alpha_k$ for constraint (\ref{eq:svmp_constraint1}) and $\beta_k$ for constraint (\ref{eq:svmp_constraint2}), the dual becomes:
\begin{align}
\max_{\boldsymbol{\alpha}, \boldsymbol{\beta}} \ & \sum_{k=1}^{N} \alpha_k - \frac{1}{2} \sum_{k,k'=1}^{N} \alpha_k \alpha_{k'} y_k y_{k'} K(\mathbf{x}_k, \mathbf{x}_{k'}) \notag \\
& \quad - \frac{1}{2 C^*} \sum_{k,k'=1}^{N} (\alpha_k + \beta_k - C)(\alpha_{k'} + \beta_{k'} - C) K^*(\mathbf{x}^*_k, \mathbf{x}^*_{k'}), \label{eq:svmp_dual}\\
\text{subject to } \ & \sum_{k=1}^{N} \alpha_k y_k = 0, \quad \sum_{k=1}^{N} (\alpha_k + \beta_k - C) = 0, \quad \alpha_k \geq 0, \quad \beta_k \geq 0,
\end{align}
where $K^*(\mathbf{x}^*_k, \mathbf{x}^*_{k'}) = \langle \psi(\mathbf{x}^*_k), \psi(\mathbf{x}^*_{k'}) \rangle$ is the kernel function in the privileged space. The second term in (\ref{eq:svmp_dual}) captures the influence of privileged information through the kernel $K^*$ in the privileged space.

\item \textit{Two-space architecture}: SVM+ operates with two distinct kernels in two different spaces. The decision function:
\begin{equation}
f(\mathbf{x}) = \text{sign}\left(\sum_{k=1}^{N} y_k \alpha_k K(\mathbf{x}_k, \mathbf{x}) + b\right)
\end{equation}
depends only on the kernel $K$ in the decision space, but the optimal values of $\alpha_k$ are influenced by similarity measures in \textit{both} spaces through the coupled optimization in (\ref{eq:svmp_dual}). The correcting function:
\begin{equation}
f^*(\mathbf{x}^*) = \frac{1}{C^*} \sum_{k=1}^{N} (\alpha_k + \beta_k - C) K^*(\mathbf{x}^*_k, \mathbf{x}^*) + b^*
\end{equation}
estimates the difficulty of examples based on their privileged features.

\item \textit{Hyperparameter complexity}: While standard SVM requires tuning 2 hyperparameters (the regularization parameter $C$ and kernel parameters for $K$), SVM+ requires tuning 4 hyperparameters: $C$ and $\gamma$ for regularization in the two spaces, plus kernel parameters for both $K$ and $K^*$.
\end{enumerate}

The key theoretical advantage of SVM+ is improved convergence rate. When the privileged information is informative, SVM+ can achieve a convergence rate of $O(\sqrt{h^*/N})$, where $h^*$ is the VC dimension of the correcting function space in the privileged domain. Since typically $h^* \ll h$ (the privileged features often live in a simpler, more structured space), this can lead to significantly faster convergence. Empirically, this manifests as reducing the number of training samples required to reach a specific accuracy from $O(N)$ to $O(\sqrt{N})$ \cite{Vapnik2009, Vapnik2015}. The intuition is that by explicitly modeling which examples are difficult through privileged features, SVM+ constructs more informed decision boundaries that account for the varying complexity of different regions in the input space.

The SVM+ framework provides a suitable algorithmic foundation for our quantum-enhanced learning approach. In our setting, the standard inputs \(\mathbf{x}_k\) represent classical features that are readily available during both training and testing, while the privileged information \(\mathbf{x}^*_k\) corresponds to quantum-derived features that capture quantum properties of the data. This alignment allows us to leverage computational capabilities of quantum devices to extract otherwise inaccessible information about the learning problem structure.

\newcommand{\cD}{\mathcal{D}}

\section{Circular Secure Decisional Diffie-Hellman}
\label{sec:crypto_stuff}

In this section we give the main cryptographic assumptions we build on relative to a deterministic group generation algorithm $\mathsf{GroupGen}$ (according to Definition~\ref{def:GroupGen}). For more discussion, we refer to Section~\ref{sec:familyofgroups}.

We note that in the following proof, we choose $a$ to be sampled as $a\gets\{0,1\}^n$ and then mapped into $\mathbb{Z}$ via $\iota(a)=\sum_{i=1}^{n} a_i \cdot 2^{n-i}$. This is solely for convenience, as is will be useful in the security reductions later. Note though that  as the two distributions are statistically close, this distinction does not affect the hardness of the underlying assumptions.

\begin{definition}[Decisional Diffie--Hellman (DDH) Assumption]
Let $\mathsf{GroupGen}$ be a deterministic group generation algorithm (according to Definition~\ref{def:GroupGen}) and let 
$\GG:=\bigl(G,g,q\bigr)\gets \mathsf{GroupGen}(1^n)$. 

We say the \emph{DDH assumption} holds relative to $\mathsf{GroupGen}$ if for any PPT adversary $\mathcal{A}$, the advantage
\begin{equation}
\left| 
\Pr[\mathcal{A}(\GG, g^{\iota(a)}, g^b, g^{\iota(a)\cdot b}) = 1] 
- \Pr[\mathcal{A}(\GG, g^{\iota(a)}, g^b, g^c) = 1]
\right|
\end{equation}
is negligible in $n$, where $\iota(a)=\sum_{i=1}^{n} a_i \cdot 2^{n-i}$ and where the probability is taken over the random choice of $a\gets\{0,1\}^n$, $b,c\gets\Z_q$ and the random coins of $\advA$. 
\end{definition}

The following definition is based on the circular power DDH assumption from \cite{cryptoeprint:2024/2073}, which was proven secure in the Generic Group Model. For our purposes, we only require a slightly simpler variant, which we refer to as \emph{circular DDH}; this assumption is implied by circular power DDH. While the circular power DDH assumption has not previously been studied in the fixed-group setting, we do not expect any classical speed-up in breaking the assumption beyond what is known for standard DDH.

\begin{definition}[Circular DDH Assumption (implied by \cite{cryptoeprint:2024/2073}, Definition 9)]
Let $\mathsf{GroupGen}$ be a deterministic group generation algorithm (according to Definition~\ref{def:GroupGen}) and let 
$\GG:=\bigl(G,g,q\bigr)\gets \mathsf{GroupGen}(1^n)$.

We say the \emph{circular DDH assumption} holds relative to $\mathsf{GroupGen}$ if for any PPT adversary $\mathcal{A}$, the advantage
\begin{equation}
\left| 
\Pr[\mathcal{A}(\GG, g^{\iota(s)},( g^{b_i}, g^{b_i\cdot \iota(s)}\cdot g^{s_i})_{i\in [n]}) = 1] 
- \Pr[\mathcal{A}(\GG, g^{\iota(s)}, ( g^{b_i}, g^{c_i})_{i\in [n]}) = 1]
\right|
\end{equation}

is negligible in $n$,  where $\iota(s)=\sum_{i=1}^{n} s_i \cdot 2^{n-i}$  and where the probability is taken over the random choice of $b_1,\dots,b_n,c_1,\dots,c_n\gets\Z_q$, $s\gets \{0,1\}^n$ and the random coins of $\advA$. 
\end{definition}
Before continuing, we show that, as expected, the circular DDH assumption implies the standard DDH assumption, since any circular DDH adversary can be used to break standard DDH. This fact will be useful later in our security proofs.

\begin{theorem} Let $\mathsf{GroupGen}$ be a deterministic group generation algorithm (according to Definition~\ref{def:GroupGen}). If the circular DDH assumptions holds relative to $\mathsf{GroupGen}$, then the DDH assumption holds relative to $\mathsf{GroupGen}$. 
\end{theorem}
\begin{proof} 

    Assume $\advA$ is a PPT adversary that breaks the DDH assumption with non-negligible advantage $\epsilon_\advA$. Then, we construct a PPT adversary $\advB$ that breaks the circular DDH assumption with non-negligible advantage $\epsilon_{\advB}=\epsilon_{\advA}/2$ as follows. On input $\GG, g^{\iota(s)},( g^{b_i}, Z_i)_{i\in [n]}$, where $Z_i=g^{b_i\cdot\iota(s)}\cdot g^{s_i}$ or $Z_i=g^{c_i}$ for all $i\in[n]$, the adversary $\advB$ flips a random bit $\beta\gets\{0,1\}$, samples a random $r\gets\Z_q$, and forwards $(\GG,(g^{\iota(s)})^r ,g^{b_1},(Z_1\cdot g^{-\beta}))^r$ to $\advA$. Note that if $Z_i=g^{b_i\cdot\iota(s)}\cdot g^{s_i}$ for all $i\in[n]$ (i.e., we are in the ``real'' circular DDH case) and $s_1=\beta$ (i.e., the adversary $\advB$ guessed the first bit of $s$ correctly), we have that $(Z_1\cdot g^{-\beta}))^r=(g^{b_1\cdot \iota(s)}\cdot g^{s_1}\cdot g^{-s_1})^r=g^{b_1\cdot \iota(s)\cdot r}$ and thus the input of $\advA$ is distributed like a ``real'' DDH tuple (for $a=rs$ and $b=b_1$). If, on the other hand, $s_1\neq \beta$, we have that $(Z_1\cdot g^{-\beta}))^r=(g^{b_1\cdot\iota(s)}\cdot g^{s_1}\cdot g^{1-s_1})^r=(g^{b_1\cdot\iota(s)}\cdot g)^r=g^{b_1\cdot\iota(s)\cdot r}\cdot g^r$. Even given $g^{\iota(s)\cdot r}$ and $g^{b_1}$, this is still distributed uniformly at random, and thus in this case the input of $\advA$ is distributed like a random tuples $(\GG,g^{\iota(a)} ,g^b,g^c)$. Altogether, we obtain
    \begin{align}
&\left| 
\Pr[\mathcal{B}(\GG, g^{\iota(s)},( g^{b_i}, g^{b_i\cdot\iota(s)}\cdot g^{s_i})_{i\in [n]}) = 1] 
- \Pr[\mathcal{B}(\GG, g^{\iota(s)}, ( g^{b_i}, g^{c_i})_{i\in [n]}) = 1]\right|
\\&=\left| \Pr[s_1=\beta]\cdot \mathcal{B}(\GG, g^{\iota(s)},( g^{b_i}, g^{b_i\cdot\iota(s)}\cdot g^{s_i})_{i\in [n]}) = 1\mid s_1=\beta] +\Pr[s_1\neq\beta]\cdot \mathcal{B}(\GG, g^{\iota(s)},( g^{b_i}, g^{b_i\cdot\iota(s)}\cdot g^{s_i})_{i\in [n]}) = 1\mid s_1\neq \beta]\right.\\&\left.- \Pr[\mathcal{B}(\GG, g^{\iota(s)}, ( g^{b_i}, g^{c_i})_{i\in [n]}) = 1]\right|
\\&=
\left| 
\frac{1}{2}\Pr[\mathcal{A}(\GG, g^{\iota(a)}, g^b, g^{\iota(a)\cdot b}) = 1] + \frac{1}{2}\Pr[\mathcal{A}(\GG, g^{\iota(a)}, g^b, g^c) = 1]-\Pr[\mathcal{A}(\GG, g^{\iota(a)}, g^b, g^c) = 1]
\right|
\\&=
\left| 
\frac{1}{2}\Pr[\mathcal{A}(\GG, g^{\iota(a)}, g^b, g^{\iota(a)\cdot b}) = 1] - \frac{1}{2}\Pr[\mathcal{A}(\GG, g^{\iota(a)}, g^b, g^c) = 1]
\right|
\\&=\epsilon_{\advA}/2. 
\end{align}
    
\end{proof}

To establish the hardness of our concept class, we rely on the following assumption, which we refer to as \emph{$Q$-time circular DDH}. 
We will show that this assumption is in fact implied by the circular DDH assumption, using standard rerandomization arguments.

\begin{definition}[$Q$-Time Circular DDH Assumption]
Let $\mathsf{GroupGen}$ be a deterministic group generation algorithm (according to Definition~\ref{def:GroupGen}) and let 
$\GG:=\bigl(G,g,q\bigr)\gets \mathsf{GroupGen}(1^n)$.

We say the \emph{$Q$-time circular DDH assumption} holds relative to $\mathsf{GroupGen}$ if for any PPT adversary $\mathcal{A}$, the advantage
\begin{equation}
\left| 
\Pr[\mathcal{A}(\GG, g^{\iota(s)},\{( g^{b_{i,j}}, g^{b_{i,j}\cdot \iota(s)}\cdot g^{s_i})_{i\in [n]}\}_{j\in [Q]}) = 1] 
- \Pr[\mathcal{A}(\GG, \{g^{\iota(s)}, ( g^{b_{i,j}}, g^{c_{i,j}})_{i\in [n]}\}_{j\in[Q]}) = 1]
\right|
\end{equation}
is negligible in $n$, where $\iota(s)=\sum_{i=1}^{n} s_i \cdot 2^{n-i}$  and where the probability is taken over the random choice of $b_{1,1,},\dots,b_{n,Q},c_{1,1},\dots,c_{n,Q}\gets\Z_q$, $s\gets \{0,1\}^n$ and the random coins of $\advA$. 
\end{definition}

\begin{theorem} Let $\mathsf{GroupGen}$ be a deterministic group generation algorithm (according to Definition~\ref{def:GroupGen}). If the circular DDH assumptions holds relative to $\mathsf{GroupGen}$, and if $Q$ is bounded by a polynomial, then the $Q$-Time circular DDH assumption holds relative to $\mathsf{GroupGen}$. 
\end{theorem}
\begin{proof} We proceed the proof in two steps. We first show that the ``real'' distribution (i.e., the distribution that is essentially an encryption of the secret), is indistinguishable from the distribution
$$(\GG, g^{\iota(s)}, ( g^{b_{i}+r_{i,j}}, g^{c_{i}+r_{i,j}\cdot \iota(s)})_{(i,j)\in [n]\times[Q]}),$$
where $r_{i,j}\gets\Z_q$ for $i\in[n]$, $j\in[Q]$. We then show that this distribution is indistinguishable from random based on the DDH assumption. 

Assume there exists an adversary $\advA$  on the $Q$-Time circular DDH assumption that wins with non-negligible advantage $\epsilon_\advA=\epsilon_\advA(n)$. Then, we construct an adversary $\advB$ on the circular DDH assumption as follows. On input $(\GG, g^{\iota(s)},( g^{b_i}, Z_i)_{i\in [n]})$, where $Z_i=g^{b_i\cdot\iota(s)}\cdot g^{s_i}$ or $Z_i=g^{c_i}$, the adversary $\advB$ samples $r_{1,1},\dots,r_{n,Q}\gets\Z_q$ and sets $C_{i,j}:=g^{b_i}\cdot g^{r_{i,j}}$ and $Z_{i,j}:=Z\cdot (g^{\iota(s)})^{r_{i,j}}$. It runs $\advA$ on input $(\mathbb{G}_n,g^{\iota(s)}, \{C_{i,j}\}_{(i,j)\in [n]\times[q]},\{Z_{i,j}\}_{(i,j)\in [n]\times[q]})$ and returns the same output. Now, if  $Z_i=g^{b_i\cdot\iota(s)}\cdot g^{s_i}$, we have that $$C_{i,j}=g^{b_i+r_{i,j}}\text{ and }Z_{i,j}=g^{b_i\cdot\iota(s)}\cdot g^{s_i}\cdot (g^{\iota(s)})^{r_{i,j}}=g^{(b_i+r_{i,j})\iota(s)}\cdot g^{s_i}.$$ Setting $b_{i,j}:=b_i+r_{i,j}$, we thus obtain that the input of $\advA$ is indeed distributed as a ``real'' $Q$-times circular DDH tuple. If $Z_i=g^{c_i}$, we have that  $$C_{i,j}=g^{b_i+r_{i,j}}\text{ and }Z_{i,j}=g^{c_i}\cdot (g^{\iota(s)})^{r_{i,j}}=g^{c_i+r_{i,j}\iota(s)}.$$ In this case, the input of $\advA$ is thus distributed according to the intermediary distribution as defined in the beginning of the proof. It is thus left to show that this distribution is indeed indistinguishable from random, if the DDH assumption is true (recall that the DDH assumption is implied by the circular DDH assumption). 

To this end, we assume there exists an adversary $\advA'$  that distinguishes the intermediary distribution from random with non-negligible advantage $\epsilon_{\advA'}=\epsilon_{\advA'}(n)$. Then, we construct an adversary $\advB'$ on the DDH assumption as follows. Given $(\GG,g^{\iota(a)},g^b,Z)$, where $Z=g^{\iota(a)\cdot b}$ or $Z=g^c$, the adversary $\advB'$ samples $b_i,c_i\gets\Z_q$ for $i\in[n]$ as well as $\alpha_{i,j},\beta_{i,j}\gets \Z_q$ for $(i,j)\in[n]\times[Q]$. The adversary sets $C_{i,j}:=g^{b_i}\cdot g^{\alpha_i}\cdot (g^b)^{\beta_i}$ and $Z_{i,j}:=g^{c_i}\cdot (g^{\iota(a)})^{\alpha_i}\cdot Z^{\beta_i}$, and runs $\advA'$ on input $(\GG,g^{\iota(a)},\{C_{i,j}\}_{(i,j)\in[n]\
times[q]},\{Z_{i,j}\}_{(i,j)\in[n]\times[q]})$. If $Z=g^{\iota(a)\cdot b}$, we have that 
$$C_{i,j}=g^{b_i+\alpha_{i,j}+b\cdot \beta_{i,j}}\text{ and }Z_{i,j}=g^{c_i+ \iota(a)\cdot \alpha_{i,j}  + \iota(a)\cdot b\cdot \beta{i,j}}=g^{c_i+\iota(a)(\alpha_{i,j}+b\cdot \beta_{i,j})}. $$ Setting $r_{i,j}:=\alpha_{i,j}+b\cdot\beta_{i,j}$ and $s:=a$, we thus obtain that the input of $\advA'$ is distributed according to the intermediary distribution. If $Z=g^c$ on the other hand, we obtain that $$C_{i,j}=g^{b_i+\alpha_{i,j}+b\cdot \beta_{i,j}}\text{ and }Z_{i,j}=g^{c_i+ \iota(a)\cdot \alpha_{i,j}  + c\cdot \beta{i,j}}.$$
Since $q$ is prime, we have that $\beta_{i,j}$ is distributed uniformly at random over $\Z_q$ even given $\iota(a)\cdot \alpha_{i,j}+\beta_{i,j}$. We thus have that $b_{i,j}:=b_i+\alpha_{i,j}+b\cdot \beta_{i,j}$ and $c_{i,j}:=c_i+ a\cdot \alpha_{i,j}  + c\cdot \beta{i,j}$ are distributed independently uniformly at random over $\Z_q$, which concludes the proof. 

\end{proof}

%


\end{document}